\documentclass[11 pt]{article}
\usepackage{hyperref}
\usepackage{graphicx,color,xcolor} 
\usepackage{amsmath,amssymb,amsfonts,amsthm}
\usepackage[title]{appendix}
\usepackage[backend=biber]{biblatex}
\newtheorem{thm}{Theorem}
\newtheorem{lemma}[thm]{Lemma}
\newtheorem{prop}[thm]{Proposition}
\newtheorem{defi}[thm]{Definition}

\newtheorem{coro}[thm]{Corollary} 

\newtheorem{assumption}[thm]{Assumption}

\usepackage{amsthm} 
\theoremstyle{remark} 
\newtheorem{remark}{Remark}

\newcommand{\Ind}{\mbox{Ind}}

\newcommand{\eps}{\varepsilon}

\renewcommand{\Re}{{\mbox{Re}}}
\renewcommand{\i}{{\rm{i}}}

\newtheorem{definition}[thm]{Definition}

\newcommand{\bbC}{\mathbb C}

\newcommand{\bbN}{\mathbb N}

\newcommand{\bbR}{\mathbb R}
\newcommand{\bbT}{\mathbb T}

\newcommand{\bc}{\mathbf c}

\newcommand{\bu}{\mathbf u} \newcommand{\bv}{\mathbf v}

 \newcommand{\bB}{\mathbf B}
  
\newcommand{\bE}{\mathbf E} 
 \newcommand{\bH}{\mathbf H}

\newcommand{\bU}{\mathbf U} \newcommand{\bV}{\mathbf V}

\newcommand{\beps}{\boldsymbol{\eps}}
\newcommand{\bmu}{\boldsymbol{\mu}}
\newcommand{\balpha}{\boldsymbol{\alpha}}
\newcommand{\bPhi}{\boldsymbol{\Phi}}

\newcommand{\cA}{\mathcal A} \newcommand{\cB}{\mathcal B}

\newcommand{\cG}{\mathcal G} 
\newcommand{\cI}{\mathcal I} \newcommand{\cJ}{\mathcal J}
 \newcommand{\cL}{\mathcal L}
\newcommand{\cM}{\mathcal M} 
  
\newcommand{\cQ}{\mathcal Q} 
 \newcommand{\cT}{\mathcal T}

\title{Spectral topology and skin modes for one-dimensional non-Hermitian photonic crystals}

 \author{
     Junshan Lin
  \thanks{Department of Mathematics and Statistics, Auburn University, Auburn, AL 36849.  \tt jzl0097@auburn.edu.}\,\,\,
 and
  Hai Zhang
  \thanks{Department of Mathematics,  HKUST,  Clear Water Bay, Kowloon, Hong Kong S.A.R., China. {\tt haizhang@ust.hk}. H.Z was partially supported by Hong Kong RGC grant GRF 16307024 and GRF 16301625.}} 

\date{}

\begin{document}

\maketitle

\begin{abstract} 
This work investigates skin modes in non-Hermitian photonic crystals with broken spectral reciprocity. In such systems, the spectra of the underlying operators can form closed loops over the complex plane with nontrivial spectral topology, which gives rise to the so-called skin effect characterized by eigenmodes localized at interfaces. For discrete lattice models, the skin effect can be understood through the spectral theory of Toeplitz matrices. However, this mathematical framework no longer applies to continuous wave models, where finite-dimensional approximations break down. In this work, we employ a transfer matrix approach to describe wave propagation in one-dimensional periodic media and introduce a new spectral topological invariant based on the eigenvalues of the transfer matrix. The new topological invariant is equivalent to the winding number of the non-Hermitian spectrum, and it enables the characterization of skin modes in one-dimensional non-Hermitian photonic crystals. The mathematical theory provides the theoretical foundation for the skin effect in continuous wave models.
\end{abstract}

\section{Introduction}
\subsection{Background}
Non-Hermitian topology has attracted significant research interest in recent years, due to its ubiquity in materials and unprecedented physics phenomena induced by the band topology. We refer the reader to the review articles \cite{ashida2020, bergholtz2021, ding2022, kawabata2019, okuma2023} and the references therein for recent theoretical developments in non-Hermitian topology, including topological phases in discrete non-Hermitian systems, anomalous bulk–edge correspondence, and exceptional points.

Compared to Hermitian systems which attain spectrum on the real line, the spectrum of non-Hermitian systems generally lies on the complex plane. As such, besides topology induced by the eigenfunctions (eigenvectors) of the operator, the eigenvalues of the non-Hermitian operator may attain topological invariant that is not present in the Hermitian operator, which could lead to new physical phenomena. 
One prominent discovery along this line of research is the so-called skin effect, which describes a macroscopic accumulation of skin (edge) modes localized on the boundary of a bulk medium under open boundary condition \cite{okuma2023, yao2018, zhang2022}.  Among many new topological invariants proposed for non-Hermitian operators, the one most relevant to the present work is the point-gap winding number, which predicts the emergence of the non-Hermitian skin effect and the accumulation of eigenmodes at the boundary.

The skin effect could be attributed to the topologically non-trivial features of the spectrum for the non-Hermitian system \cite{okuma2020, zhang2019}. For the tight-binding model, the relation between the spectral topology of the infinite periodic problem and the spectrum of the semi-infinite bulk medium is well understood through the spectral theory of the Toeplitz matrices \cite{okuma2020}. More precisely, the infinite periodic operator is a Laurent operator $L_p$ and the corresponding semi-infinite operator is a Toeplitz operator $T$. If the spectrum $\sigma(L_p)$ forms a closed curve on the complex plane with a nonzero winding number, then the spectrum $\sigma(T)=\sigma(L_p) \cup \sigma_{wind}$ (cf. \cite{bottcher2005}), where 
\begin{equation}\label{eq:sigma_T}
   \sigma_{wind} :=\{ \lambda\in \bbC\backslash \sigma(L_p): wind(\sigma(L_p)-\lambda) \neq 0   \}.
\end{equation}
In the above, $\sigma_{wind}$ corresponds to the spectrum of the skin modes, which is the area enclosed by the dispersion curve $\sigma(L_p)$. The spectral theory of Toeplitz operators has been further developed to study the non-Hermitian skin effect for continuous wave models \cite{ammari2024, ammari2025}. Therein, the non-Hermitian system consists of an array of subwavelength resonators with an imaginary gauge potential.
When the material contrast is high, asymptotic analysis shows that the underlying differential-equation model can be effectively reduced to a finite-dimensional system governed by the so-called capacitance matrix, which possesses a Toeplitz structure. In this work, we study the non-Hermitian skin effect in continuous wave models, for which finite-dimensional approximations break down, and the spectral theory of Toeplitz operators no longer applies.

\subsection{Main results}
We consider non-Hermitian photonic crystals with material loss and broken spectral reciprocity (see Definition \ref{def1}). Numerical studies in \cite{fang2022, zhong2021} indicate that the eigenvalues corresponding to the skin modes of the finite-size photonic crystal accumulate within the region enclosed by the spectrum $\sigma(L_p)$ of the infinite periodic operator, thereby manifesting the photonic skin effect. We aim to develop a rigorous mathematical framework that elucidates the spectral topology of non-Hermitian photonic crystals and the emergence of skin modes in semi-infinite photonic structures. 

Our main results show that the spectrum of skin modes is characterized by $\sigma_{wind}$, which admits the representation \eqref{eq:sigma_T}. The analysis rests on a transfer-matrix formulation of wave propagation in one-dimensional periodic media. Alongside the classical winding number of the spectral curve of the periodic non-Hermitian operator, we introduce a new spectral topological invariant defined through the eigenvalues of the transfer matrix. 
For simple bounded components over the complex plane, these two invariants coincide. Furthermore, the newly introduced invariant precisely characterizes the magnitude of the transfer matrix eigenvalues,
which allows for the characterization of the edge modes in terms of the corresponding eigenvectors.

\medskip

We summarize the problem setting and main results as follows:

\medskip

\noindent\emph{Problem setting.} We consider electromagnetic wave propagation in layered medium with periodic permittivity and permeability tensors $(\beps(z),\bmu(z))$ along the $z$ direction. The transverse field $\bPhi(z):=(E_x(z), E_y(z), H_x(z), H_y(z))^T\in\bbC^4$ is governed by the Maxwell's system \eqref{eq:ODE2}.

\medskip

\noindent\emph{The topological invariant.} 
The propagation of wave field $\bPhi(z)$ over one period can be described by the monodromy matrix $\cM(\omega)\in\bbC^{4\times4}$ (or transfer matrix starting on the left of the period; see \eqref{eq:monodomy_mat} for its definition).
For a connected component $D$ on the complex plane that does not intersect with the dispersion curves, we define the topological index 
\[
\Ind(D)=\tfrac12\big(\#\{j:|\lambda_j(\omega)|<1\}-\#\{j:|\lambda_j(\omega)|>1\}\big), \quad  \omega\in D,
\]
where the $\lambda_j(\omega)$ are the eigenvalues of $\cM(\omega)$. It is shown that $\Ind(D)$ is well defined and it attains an integer values $0$ or $\pm 1$.

\medskip

\noindent\emph{Main results.} The main results are given in Theorems~\ref{lem-0},~\ref{thm:ind_homotopy}, \ref{thm:Ind_Dn}, and \ref{thm:edge_modes}.
\begin{itemize}
\item[(i)] Under a spectral-gap hypothesis for a Hermitian medium, we prove that the complement $\Gamma^c$ of its spectrum in the complex plane has index $0$ (cf. Theorem~\ref{lem-0}). We further show that this index is invariant under homotopy within the admissible coefficient class (cf. Theorem~\ref{thm:ind_homotopy}).

Specifically, assuming that the non-Hermitian medium can be continuously deformed from a Hermitian one possessing a spectral gap, we characterize when a connected component $D$ of the
complement of the dispersion curves in the complex plane is a principal gapped component for the non-Hermitian medium attaining $\Ind(D)=0$.

\item[(ii)] Let \(D_n^i\) be the Jordan domain enclosed by all or part of a closed dispersion curve \(\gamma_n\), where
\(\gamma_n\) is the image of a globally defined spectral band
\(\omega_n:\mathbb{T}\to\mathbb{C}\). If \(D_n^i\) is adjacent to a principal gapped component \(D\) satisfying \(\Ind(D)=0\), then under mild regularity assumptions on the dispersion curve, we prove that the topological index $\Ind(D_n^i)=\pm 1$, where the sign is determined by the winding direction of the dispersion curve. (cf. Theorem~\ref{thm:Ind_Dn}).

\item[(iii)] Let \(D_n^i\) be a Jordan domain bounded by all or part of a closed dispersion curve \(\gamma_n\) which is the image of a globally defined spectral band
\(\omega_n:\mathbb{T}\to\mathbb{C}\). Assume that its topological index $\Ind(D_n^i)=\pm 1$. We prove that each $\omega \in D_n^i$ is an eigenvalue associated with a skin mode (cf. Theorem~\ref{thm:edge_modes}).

\end{itemize}

\subsection{Mathematical model for the one-dimensional photonic crystals}
We consider a periodic layered medium with the permittivity and permeability tensors varying along the $z$ direction, which are given by
\begin{equation}
   \beps(z) = \left(
  \begin{matrix}
       \eps_{xx}(z) & \eps_{xy}(z) & 0 \\
       \eps_{yx}(z) & \eps_{yy}(z) & 0 \\
         0  & 0 & \eps_{zz}(z)  
  \end{matrix}
   \right), 
   \quad
  \bmu(z) = \left(
  \begin{matrix}
       \mu_{xx}(z)  & \mu_{xy}(z) & 0 \\
       \mu_{yx}(z)  & \mu_{yy}(z) & 0 \\
         0  & 0 & \mu_{zz}(z)  
  \end{matrix}
   \right).
\end{equation}
Each component of $\beps(z)$ and $\bmu(z)$ is bounded and piecewisely continuous with the period $p=1$ such that
\begin{equation*}
    \beps(z+1)=\beps(z) \quad \mbox{and} \quad \bmu(z+1)=\bmu(z).
\end{equation*}

Consider a time-harmonic electromagnetic wave $\{\bE, \bH \}$
propagating along the $z$ direction with 
\begin{equation}\label{eq:EH}
    \bE = [ \, E_x(z), E_y(z), 0 \, ]^T \, e^{-\i\omega t} \quad\mbox{and}\quad \bH = [ \, H_x(z), H_y(z), 0 \,]^T \, e^{-\i\omega t}.
\end{equation}
The Maxwell's equations read as 
\begin{equation}\label{eq:Maxwell}
    \nabla \times \bE = \frac{\i \, \omega}{c} \bmu \bH, \quad \nabla \times \bH = - \frac{\i \, \omega}{c} \beps \bE.
\end{equation}

To simply the notation and the analysis, here and henceforth, we normalize the wave speed to $c=1$. By a slight abuse of notation, we continue to use $\{\bE, \bH \}$ to denote the transverse field components of the electromagnetic field by ignoring the $z$ component in \eqref{eq:EH} unless specified. Similarly, we use $\beps(z)$ and $\bmu(z)$ to denote the $2\times2$ transverse blocks corresponding to the $xy$-components of the permittivity and permeability tensors. More specifically, we introduce the following notations
\begin{equation*}
    \bPhi(z) := \left(\begin{matrix}
        \bE(z) \\ \bH(z)
    \end{matrix} \right) \in\bbC^4,
    \quad
    Q = \left(
\begin{matrix}
   0 & -1 \\
   1 & 0 
\end{matrix}
\right), \quad
    \cQ := \left(\begin{matrix}
         Q & 0 \\ 0 & Q
    \end{matrix} \right), 
    \quad
    \cA(z) := \i \, \left(\begin{matrix}
         0 & \bmu(z) \\  -\beps(z) & 0
    \end{matrix} \right). 
\end{equation*}
Then the full $6\times 6$ Maxwell system in one-dimensional medium can be reduced to the following $4\times 4$ ODE system
\begin{equation}\label{eq:ODE2}
    \frac{d}{dz} \bPhi(z) = \omega \cQ^{-1}\cA(z) \bPhi(z),
\end{equation}
where $\cA(z+1)=\cA(z)$.

If the permittivity and permeability tensors are positive definite matrices satisfying $\beps(z) = \beps^*(z)$, $\bmu(z) = \bmu^*(z)$, where $*$ denotes the conjugate transpose operation, the photonic crystal is said to be Hermitian. Otherwise, we call the photonic crystal non-Hermitian, which are typically made of lossy materials. The focus of this work is on non-Hermitian photonic crystals with broken spectral reciprocity. For such configurations, the dispersion curves may form closed loops over the complex plane and enclose nontrivial regions.

\bigskip

\begin{remark}
For clarity, we assume that $\beps$ and $\bmu$ are independent of $\omega$ so that the problem \eqref{eq:ODE2} is linear in $\omega$. This assumption is not essential for the mathematical techniques being used in this work, as the transfer-matrix construction in Section \ref{sec:band_structure_PC}, the topological index $\Ind(D)$ introduced in Section \ref{sec:top_idx_spec_gap} and the homotopy invariance of the topological index are independent of the form of $\beps$ and $\bmu$. However, for dispersive media $(\beps(z; \omega),\bmu(z; \omega))$,
the spectral theory of the nonlinear eigenvalue problem \eqref{eq:ODE2} is much more complicated than the linear one, and additional hypotheses need to be imposed for the main results to be valid. We do not pursue the study on such configuration here. 
\end{remark}

\subsection{Organization of the paper}
The rest of the paper is organized as follows. In Section 2, we study the band structure of the photonic crystals and relate the eigenvalues of the underlying differential operator with the eigenvalues of the corresponding transfer matrix. The concepts of spectral reciprocity and point gaps will be illustrated through detailed discussions. In Section 3,
we investigate the skin modes in semi-infinite photonic crystals when dispersion curves of the corresponding infinite periodic structures attain point gaps. It is shown that, for regular simple point-gap components, the winding number of the corresponding dispersion loop is equivalent to the new topological invariant defined via the eigenvalues of the underlying transfer matrix. Furthermore, we prove that the eigenfrequencies for the skin modes can be characterized by the winding number of the dispersion curves and the new topological invariant.

\section{Band structure of photonic crystals}\label{sec:band_structure_PC}
\subsection{General band theory and transfer matrix}
For each $k\in [-\pi, \pi]$, we consider the following eigenvalue problem for $z \in \bbR$:
\begin{subequations}
\label{eq:eig_prob1}
\begin{align}
\label{eq:ODE_model2}
&  \frac{d}{dz} \bPhi(k; z) = \omega \cQ^{-1}\cA(z) \bPhi(k; z),  \\ 
\label{eq:quasi_bnd}
& \bPhi(k; z+1) = e^{\i k} \bPhi(k; z).
\end{align}
\end{subequations}
In the above, the interval $\cB:=[-\pi, \pi]$ is called the Brillouin zone, $k$ is called the Bloch wavenumber,  $\omega$ is the eigenfrequency, and the corresponding eigenmode $\bPhi(k; z)$ is called the Bloch mode.

The solution of the ODE system \eqref{eq:ODE_model2} can be expressed via the transfer matrix. For each fixed $\omega\in\bbC$ and $z_1\in\bbR$, let $\cT_\omega(z,z_1)\in\bbC^{4\times4}$ be the transfer matrix associated with \eqref{eq:ODE_model2} that solves
\begin{equation}\label{eq:dT_dz}
\frac{d}{dz}\cT_\omega(z,z_1)=\omega\cQ^{-1}\cA(z)\cT_\omega(z,z_1),\qquad \cT_\omega(z_1,z_1)=\cI.    
\end{equation}
Then the solution of \eqref{eq:ODE_model2} at $z=z_2$ is related to the solution at $z=z_1$ by the following relation:
\begin{equation}\label{def:transfer_mat}
    \bPhi(z_2) = \cT_\omega(z_2, z_1) \bPhi(z_1). 
\end{equation}

\begin{lemma}
Let $\cT_\omega(z,0)$ be the transfer matrix 
that is defined by \eqref{eq:dT_dz} with $z_1=0$. Then $\mbox{det}(\cT_\omega(z,0))=1$ for any $z\in\bbR$.
\end{lemma}

\begin{proof}
    Using the Liouville's formula (cf. \cite{chicone2006}), we have
    \begin{equation*}
        \mbox{det}(\cT_\omega(z,0))= \mbox{det}(\cT_\omega(0,0)) \left( \exp{\int_0^z \mbox{tr}(\cQ^{-1}\cA(s))} \, ds \right) = \mbox{det}(\cT_\omega(0,0)) = 1,
    \end{equation*}
   where we have used the fact that $\mbox{tr}(\cQ^{-1}\cA(s))=0$.
\end{proof}

The eigenfrequency $\omega$ of \eqref{eq:eig_prob1} can be determined using the transfer matrix and the quasi-periodic boundary condition \eqref{eq:quasi_bnd}. More precisely, let 
\begin{equation}\label{eq:monodomy_mat}
   \cM(\omega):=\cT_\omega(1,0) 
\end{equation}
be the monodromy matrix. Then for each $k\in [-\pi, \pi]$, $\omega$ is an eigenfrquency of \eqref{eq:eig_prob1} if and only if $e^{\i k}$ is an eigenvalue of the matrix $\cM(\omega)$
 such that  
 \[\cM(\omega)  \bv = e^{\i k} \bv, \] 
 where $\bv \in \bbC^4$ is the corresponding eigenvector for the matrix $\cM(\omega)$. Meanwhile, the Bloch mode for  \eqref{eq:eig_prob1} with the eigevanfrequency $\omega$ at $z=0$ takes the value $\bPhi(k; 0)=\bv$. As such the dispersion relation $\omega(k)$ is determined by the following characteristic equation for the momodromy matrix $\cM(\omega)$:
 \begin{equation}\label{eq:chara1}
     \mbox{det}(\cM(\omega) - e^{\i k} I) = 0.
 \end{equation}
Setting $\lambda=e^{\i k}$, the above characteristic equation reads
\begin{equation}\label{eq:chara2}
 P(\omega):=\lambda^4 + a_3(\omega) \lambda^3 + a_2(\omega) \lambda^2 + a_1(\omega) \lambda + 1 = 0
\end{equation}
for some coefficients $a_j(\omega)$ $(j=1, 2, 3)$, which are all analytic in $\omega$.

\begin{lemma} \label{lem2}
 For each complex-valued $\omega\in\bbC$, there exist four roots to the characteristic equation \eqref{eq:chara2} denoted by $\lambda_i(\omega)$ ($i=1, 2, 3, 4$). Each function $\lambda_i(\omega)$ is locally analytic except at branch points.  More precisely, for any $\omega_0$,  $\lambda_i(\omega)$ is analytic near $\omega_0$ if 
 \[
 \lambda_i(\omega_0) \neq \lambda_j(\omega_0), \quad \mbox{for all}\,\,j \neq i.
 \]
Moreover, $\lambda_{1}\lambda_{2}\lambda_{3}\lambda_{4} =1$. 
\end{lemma}

Lemma \ref{lem2} follows from the analytic dependence of the eigenvalues of an analytic matrix family on the parameter; see, for instance,  \cite{kato1966}. In view of the above discussions, we have the following proposition that characterizes the eigenfrequencies of \eqref{eq:eig_prob1} using the monodromy matrix.

\begin{prop}\label{prop:rel_lambda_omega}
For each $k\in\cB$, the eigenfrequencies of \eqref{eq:eig_prob1} are precisely those $\omega\in\bbC$ for which the monodromy matrix $\cM(\omega)$ has an eigenvalue equal to $e^{\i k}$. More precisely, Let $\Omega(k)$ denote the set of eigenfrequencies of \eqref{eq:eig_prob1}, then
\begin{equation}\label{eq:rel_lambda_omega}
   \Omega(k)=\{\omega\in\bbC: \lambda_j(\omega) = e^{\i k} \,\,\mbox{for some} \,\, 1\leq j \leq 4\}. 
\end{equation}
\end{prop}

\subsection{Band structure of the Hermitian photonic crystals}
\subsubsection{Floquet-Bloch theory and the monodromy matrix}
In this section, we introduce the Band structure for one-dimensional Hermitian photonic crystals and analyze the associated monodromy matrices.  

We start from layered Hermitian media in which the permittivity and permeability tensors are positive-definite matrices satisfying $\beps(z) = \beps^*(z)$ and $\bmu(z) = \bmu^*(z)$.  By eliminating the magnetic field $\bH$ in the Maxwell's equations \eqref{eq:Maxwell}, we obtain the following second-order Maxwell system
\begin{equation}\label{eq:Maxwell_sys}
    \nabla \times (\bmu^{-1} \nabla \times \bE) - \omega^2 \beps \bE = \bf{0}.
\end{equation}
Using only the transverse electric field $\bE(k;z)=(E_x,E_y)^T\in\bbC^2$ and the transverse $2\times2$ blocks of the tensors $\beps(z),\bmu(z)$,  we obtain the following eigenvalue problem that is equivalent to \eqref{eq:eig_prob1}:
\begin{subequations}
\label{eq:eig_Maxwell}
\begin{align}
& \partial_z \Big( \bmu(z)^{-1} Q \partial_z \bE \Big) + \omega^2 Q \beps(z) \bE = 0,  \\ 
& \bE(k; z+1) = e^{\i k} \bE(k; z).
\end{align}
\end{subequations}
Introduce the operator \(\mathcal{L} = -Q \partial_z \bmu(z)^{-1} Q \partial_z\). 
Then the above eigenproblem can be transformed to the following generalized self-adjoint eigenvalue problem:
\[
\mathcal{L} \bE = \omega^2 \beps(z) \bE, \quad \bE(k; z+1) = e^{\i k} \bE(k; z),
\]
where \(\beps(z)\) acts as a positive-definite Hermitian weight matrix. It follows from the standard spectral theory for the self-adjoint second-order elliptic differential operator that \eqref{eq:eig_Maxwell}
attains a sequence of real eigenvalues ordered in the following way:
\begin{equation*}
    \nu_1(k) \le \nu_2(k) \le \cdots \le \nu_n(k) \le \nu_{n+1}(k) \le \cdots.
\end{equation*}
We focus on the non-negative eigenfrequencies $\omega_n:=\sqrt{\nu_n}$ for \eqref{eq:eig_prob1}, which is ordered as 
\begin{equation*}
    \omega_1(k) \le \omega_2(k) \le \cdots \le \omega_n(k) \le \omega_{n+1}(k) \le \cdots.
\end{equation*}

For each $n$, the dispersion relation $\omega_n(k)$ is a continuous function for $k\in \cB$. If $\max_{k\in\cB}\omega_n(k) < \min_{k\in\cB}\omega_{n+1}(k)$, we say there is a spectral gap between the band $\omega_n(k)$ and $\omega_{n+1}(k)$.
The collection of the countable set of dispersion relations $\{\omega_n(k)\}_{n=1}^\infty$ for $k\in \cB$ is called the band structure of the periodic medium. From the Floquet-Bloch theory \cite{kuchment2012floquet}, the spectrum of the Maxwell operator $\cL$ defined in \eqref{eq:Maxwell_sys} is given by $\sigma(\cL)=\displaystyle{\cup_{k\in\cB} (\cup_{n=1}^\infty \omega_n(k)})$. The spectrum $\sigma(\cL)$ is a collection of closed intervals on the real line, which are separated by the spectral gap intervals $(\max_{k\in\cB}\omega_n(k), \min_{k\in\cB}\omega_{n+1}(k))$ for $n=1, 2, 3, \cdots$. By virtue of Proposition \ref{prop:rel_lambda_omega}, 
$\sigma(\cL)$ can be equivalently described by the set
\[
\Omega := \{\omega\in \bbR: \lambda_j(\omega) \in \mathbb{T} \,\,\mbox{for some} \,\, 1\leq j \leq 4\},
\]
where $\mathbb{T}= \{ e^{\i k}; k \in (-\pi,\pi] \}$ is the unit circle on the complex plane.

For the Hermitian photonic crystals, the eigenvalues $\lambda_j(\omega)$ ($1\le j \le 4$) of the monodromy matrix $\cM(\omega)$ obey an additional symmetry when the frequency $\omega$ is real.  To explore the relation between these eigevanlues, we introduce the matrix
\[
\cJ=\begin{pmatrix}0 & Q^* \\ Q & 0\end{pmatrix}\in\bbC^{4\times4}.
\]
Then the quadratic form
$\tfrac12\,\bPhi^*\cJ\,\bPhi=\Re\!\big(E_x\overline{H_y}-E_y\overline{H_x}\big)$ represents the time-averaged
longitudinal Poynting flux. 
When the medium coefficient matrices $\beps(z)$ and $\bmu(z)$ are Hermitian, a straightforward calculation shows that the generator in the ODE system \eqref{eq:ODE_model2} satisfies
\begin{equation}\label{eq:flux-id}
\cJ\,\big(\cQ^{-1}\cA(z)\big)+\big(\cQ^{-1}\cA(z)\big)^*\cJ=0
\qquad\text{for a.e. }z\in\bbR,
\end{equation}
which implies that $\tfrac{d}{dz}\big(\bPhi^*\cJ\bPhi\big)=0$ for the solution of \eqref{eq:ODE_model2} with real $\omega$.

\begin{lemma}\label{lem:flux}
Assume $\beps(z) = \beps^*(z)$ and $\bmu(z) = \bmu^*(z)$,  and let $\omega\in\bbR$. Then the following holds:
\begin{itemize}
    \item [(i)] The monodromy matrix $\cM(\omega)$ satisfies  $\cM(\omega)^*\,\cJ\,\cM(\omega)=\cJ$.
    \item [(ii)] If $\lambda$ is an eigenvalue of $\cM(\omega)$, then $1/\bar\lambda$ is also an eigenvalue of $\cM(\omega)$ with the same algebraic multiplicities.
\end{itemize}
\end{lemma}

\begin{proof}
Let $\cG(z):=\omega\,\cQ^{-1}\cA(z)$. Then in view of  \eqref{eq:dT_dz}, there holds
$$ \frac{d}{dz}\cT_\omega(z,0)=\cG(z)\,\cT_\omega(z,0),  \; \cT_\omega(0,0)=\cI. $$ 
Since $\omega\in\bbR$,
the identity \eqref{eq:flux-id} gives $\cG(z)^*\cJ+\cJ\cG(z)=0$ for a.e.\ $z$. Therefore,
\[
\frac{d}{dz}\Big(\cT_\omega(z,0)^*\,\cJ\,\cT_\omega(z,0)\Big)
=\cT_\omega(z,0)^*\big(\cG(z)^*\cJ+\cJ\cG(z)\big)\cT_\omega(z,0)=0,
\]
as such $\cT_\omega(z,0)^*\cJ\,\cT_\omega(z,0)$ is constant in $z$. By evaluating $\cT_\omega(z,0)^*\cJ\,\cT_\omega(z,0)$ at $z=1$ and $z=0$, respectively, we obtain
\begin{equation}\label{eq:id_MJ}
 \cM(\omega)^*\cJ\,\cM(\omega)  =\cJ.    
\end{equation}

Recall that $\det\cM(\omega)=1$, the relation \eqref{eq:id_MJ} may be rewritten as
\begin{equation}\label{eq:id_MJ_rewritten}
    \cJ\,\cM(\omega)=\big(\cM(\omega)^*\big)^{-1}\cJ. 
\end{equation}
Let $(\lambda, \bv)$ with $\bv\in\bbC^4\setminus\{0\}$ be an eigenpair of $\cM(\omega)$ such that 
$\cM(\omega)\,\bv=\lambda\,\bv$.  Applying $\cJ$ and using the
relation \eqref{eq:id_MJ_rewritten}, we obtain
\[
\big(\cM(\omega)^*\big)^{-1}(\cJ\bv)=\cJ\big(\cM(\omega)\bv\big)=\lambda\,(\cJ\bv).
\]
It follows that $\lambda$ is an eigenvalue of $\big(\cM(\omega)^*\big)^{-1}$, 
which further implies that $1/\bar\lambda$ is also an eigenvalue of $\cM(\omega)$.
Since $\cJ$ is invertible, it maps the 
eigenspace of $\lambda$ isomorphically onto that of $1/\bar\lambda$. Hence the eigenvalues $\lambda$ and $1/\bar\lambda$ share the same algebraic multiplicity.
\end{proof}

\subsubsection{Hermitian photonic crystals with spectral reciprocity}

\begin{defi} \label{def1}
A photonic crystal is said to attain \textbf{spectral reciprocity} if $\omega_n(k)=\omega_n(-k)$ holds for any $k\in\cB$ and $n\in\bbN^+$. Otherwise, the spectrum of the photonic crystal is called non-reciprocal. 
\end{defi} 
We first consider the Hermitian photonic crystal for which the system \eqref{eq:eig_prob1} attains the spectral reciprocity. In general, the spectral reciprocity holds when the symmetry group of the period medium attains a certain operation $g$ such that $g \bPhi(k,\cdot) =  \bPhi(-k,\cdot)$. \\

\noindent Example 1. (spatial inversion symmetry): $\beps(z) = \beps(-z)$, $\bmu(z) = \bmu(-z)$: Let $k\in\cB$ and $(\omega, \bPhi(k;z))$ be an eigenpair of \eqref{eq:eig_prob1}, then it can be verified that $\tilde\bPhi(k;z):=\big( \bE(-k; -z), -\bH(-k; -z) \big )^T$ satisfies \eqref{eq:ODE_model2} with the boundary condition $\tilde\bPhi(k; 1) = e^{\i k} \tilde\bPhi(k; 0)$. Hence $(\omega, \tilde\bPhi(k;z))$ is an eigenpair of \eqref{eq:eig_prob1}.
\\

\noindent Example 2. (time-reversal symmetry): $\beps(z) = \overline{\beps(z)}$ and $\bmu(z) = \overline{\bmu(z)}$. Time-reversal symmetry says that if $\{ \bE(z,t), \bH(z,t) \}$ is a solution of the time-dependent Maxwell's equations, so is $\{ \bE(z,-t), -\bH(z,-t) \}$. 
Note that for the time-harmonic electromagnetic field that satisfies \eqref{eq:eig_prob1}, using the quasi-periodic condition $\bPhi(k;z+1)=e^{\i k}\bPhi(k;z)$, the Bloch mode can be expressed as $\bPhi(k;z)=e^{\i kz}\bU(k;z)$, where
\begin{subequations}
\label{eq:eig_prob2}
\begin{align}
 &  \left( \frac{d}{dz}  + \i k \right) \bU(k; z) = \omega \cQ^{-1}\cA(z) \bU(k; z),  \\
& \bU(k; 1) = \bU(k;0).
\end{align}
\end{subequations}
Let  $(\omega, \bU(k;z))$ be an eigenpair of \eqref{eq:eig_prob2}, where 
$\bU(k; z)=[\bu_1(k;z), \bu_2(k;z)]^T$.
Then the time-reversal symmetry implies that  $(\omega, \bV(-k;z))$ is an eigenpair of \eqref{eq:eig_prob2} for $-k$, where $\bV(-k; z)=[ \overline{\bu_1(k;z)}, -\overline{\bu_2(k;z)}]^T$.

\subsubsection{Hermitian photonic crystals breaking the spectral reciprocity}

From the discussion in the previous subsection, in order to break the spectral reciprocity of the photonic crystal, one needs to break the spatial inversion symmetry and time-reversal symmetry of the system. The former can be achieved by using more than two layers with different medium parameters, and the latter can be achieved by using magnetic photonic crystals  \cite{figotin2001}.

To this end, we introduce the so-called $A$ and $F$ layers as follows. The $A$ layers are made of a non-magnetic dielectric material with anisotropy over the $xy$ plane:
\begin{equation}\label{eq:para_A}
  \beps_A(\delta, \varphi) = \left(
  \begin{matrix}
       \eps_0 + \delta \cos(2\varphi) &  \delta \sin(2\varphi)  & 0 \\
        \delta \sin(2\varphi) & \eps_0 - \delta \cos(2\varphi)  & 0 \\
         0  & 0 & \eps_{zz}
  \end{matrix}
   \right), 
   \quad
  \bmu_A = I.
\end{equation}
The tensor $\beps_A$ is real valued, where the parameter $\delta$ describes the magnitude of in-plane anisotropy, while the angle $\varphi$ defines the orientation of the principal axes
of tensor $\beps_A$ on the $xy$ plane. The $F$ layers are ferromagnetic with magnetization parallel to the $z$ direction, for which the permittivity and permeability tensors are given by
\begin{equation}\label{eq:para_F}
   \beps_F(\alpha) = \left(
  \begin{matrix}
        \tilde\eps_0  & \i \alpha & 0 \\
        -\i \alpha & \tilde\eps_0 & 0 \\
         0  & 0 & \eps_{zz}
  \end{matrix}
   \right), 
   \quad
  \bmu_F(\beta) = \left(
  \begin{matrix}
       1 & \i \beta & 0 \\
     - \i \beta & 1 & 0 \\
         0  & 0 & \mu_{zz}
  \end{matrix}
   \right).
\end{equation}
The real parameters $\alpha$ and $\beta$ are responsible for the magnetic Faraday rotation.
It is clear that $\beps_F(\alpha)\neq\overline{\beps_F(\alpha)}$ and $\bmu_F(\beta)\neq\overline{\bmu_F(\beta)}$.

The spectral asymmetry can be achieved by stacking the $A$ and $F$ layers together. 
Let us consider a periodic medium for which each periodic cell consists of one $F$ layer sandwiched by two $A$ layers with different rotations $\varphi_1$ and $\varphi_2$. The corresponding medium parameters for the primitive cell are given by 
\begin{equation}\label{eq:para_three_layers}
    \beps(z) = \left\{ 
\begin{array}{lll}
     \beps_A(\delta, \varphi_1),  & 0<z<L;   \\
     \beps_F(\alpha),            & L<z<1-L; \\
     \beps_A(\delta, \varphi_2),  & 1-L < z < 1.
\end{array}    
    \right.
\quad    
\bmu(z) = \left\{ 
\begin{array}{lll}
     \bmu_A,           &       0<z<L;   \\
     \bmu_F(\beta),    & L<z<1-L; \\
     \bmu_A,  & 1-L < z < 1.
\end{array}    
    \right.    
\end{equation}
It the above, it is assumed that $\varphi_1-\varphi_2 \neq 0, \frac{\pi}{2}$, and $\alpha,$ $\beta \neq 0$. The periodic medium does not attain the spatial inversion symmetry by using three layers with different medium parameters.
The time-reversal symmetry is also broken by using the F layer. Hence, the corresponding Maxwell's operator may achieve spectral non-reciprocity $\omega_n(k)\neq\omega_n(-k)$. To illustrate the spectral asymmetry, let us set $\eps_0=13$, $\tilde\eps_0=1$ in \eqref{eq:para_A} and \eqref{eq:para_F}, and consider the following example. \\

\noindent Example 3. Consider a three-layer periodic medium with the permittivity and permeability tensors in the form of \eqref{eq:para_three_layers}, with the parameters $\delta=6$, $\varphi_1=0$, $\varphi_2=0.8$, $\alpha=\beta=0.5$. The corresponding monodromy matrix is expressed as a composition of the transfer matrices for the $AFA$ layer as $\cM(\omega)=\cT_A(\omega,\varphi_1;L) \cT_F(\omega;1-{\color{red}2L})\cT_A(\omega, \varphi_2;L)$, where the explicit expression of the $A$ layer and $F$ layer are given in Appendix \ref{sec:transfer_mat}. Solving the characteristic equation \eqref{eq:chara2} for each $k\in\cB$ yields the band structure of the periodic medium as shown in Figure \ref{fig:band_three_layer} (Left).

\begin{figure}[!htbp]
    \centering
    \includegraphics[width=8cm]{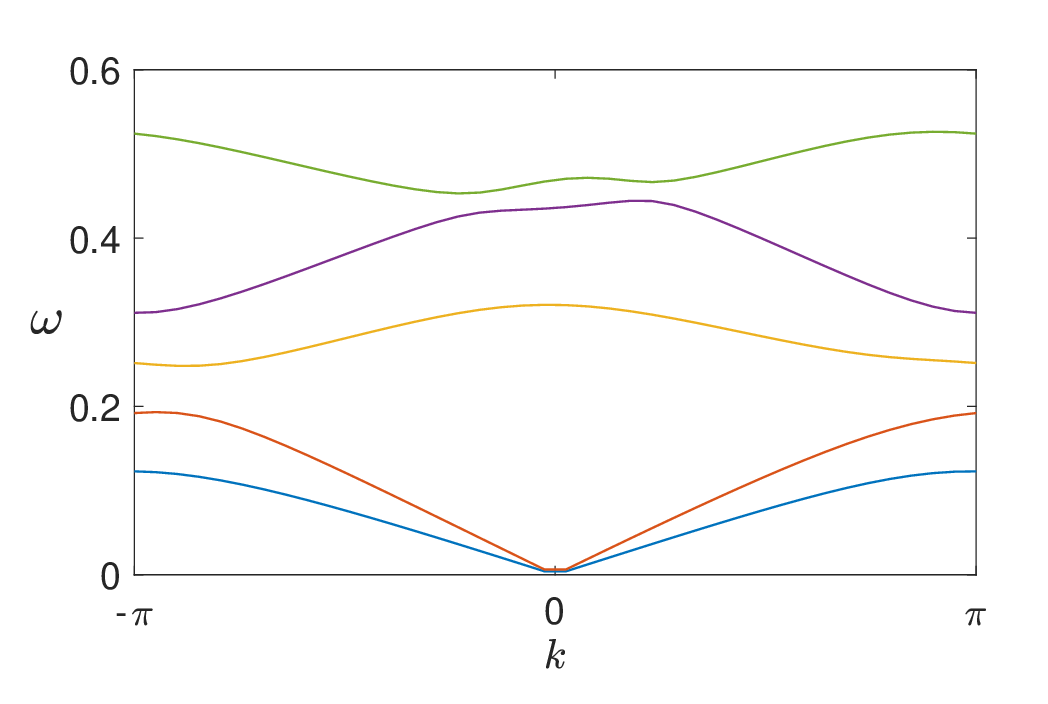}
    \includegraphics[width=7.5cm]{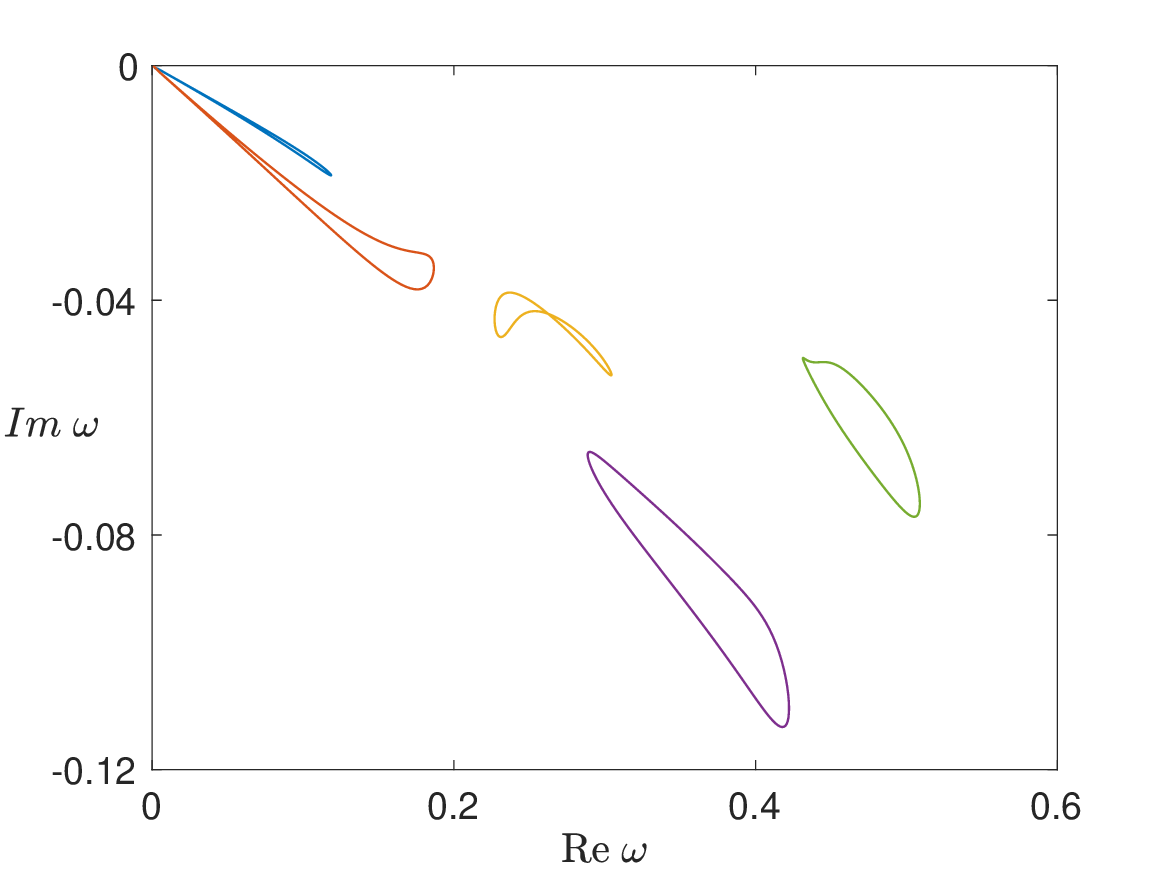}
    \caption{Band structures of the three-layer periodic media which attain spectral asymmetry. The permittivity and permeability are given in the form of \eqref{eq:para_three_layers}, with $\delta=6$, $\varphi_1=0$, $\varphi_2=0.8$, $\alpha=\beta=0.5$.
    Left: Hermitian photonic crystal with $\eps_0=13$, $\tilde\eps_0=1$;  Right: Non-Hermitian photonic crystal with $\eps_0=13+5\i$, $\tilde\eps_0=1$.  }
    \label{fig:band_three_layer}.
\end{figure}

\subsection{Band structure of non-Hermitian photonic crystals}\label{sec:band_non-Hermitian}
When $\beps(z)\neq \beps^*(z)$ or $\bmu(z)\neq \bmu^*(z)$, as is the case, for instance, when the photonic crystal is composed of lossy materials, the eigenfrequencies of \eqref{eq:eig_prob1} may become complex-valued for each $k\in\cB$. By Proposition~\ref{prop:rel_lambda_omega}, the spectral set $\Omega(k)$ can equivalently be characterized in terms of the monodromy matrix through \eqref{eq:rel_lambda_omega}. With the the family of spectral sets $\{\Omega(k): k\in\cB \}$, we define the spectral bands $\{\omega_n(k)\}_{n=1}^\infty$ such that the following are satisfied:
\begin{itemize}
    \item[(i)] Each spectral dispersion relation $\omega_n(k)$ is a continuous function of $k$;
    \item[(ii)] $\min_{k\in\cB} \; \Re \, \omega_n(k)\le \min_{k\in\cB} \Re\; \omega_{n+1}(k)$ and $\max_{k\in\cB} \Re \; \omega_n(k)\le \max_{k\in\cB} \Re \; \omega_{n+1}(k)$.
\end{itemize}
We consider the photonic crystal $(\beps(z),\bmu(z))$ for which there is no ambiguity to define the dispersion curves $\{\omega_n(k)\}_{n=1}^\infty$ with the above criteria.

\medskip

\begin{remark}
The spectral dispersion relation $\omega_n(\cdot)$ defined above is analytic at $k$ for which $\omega_n(k)$ is a simple root of the characteristic equation \eqref{eq:chara1}. It should also be noted that the labeling of the spectral bands $\{\omega_n(k)\}_{n=1}^\infty$ is not unique. Alternatively, the spectral bands may be defined as the analytic branches of the dispersion relation \eqref{eq:chara1}. We refer to Appendix~\ref{sec:def_spec_bands} for details. 
\end{remark}

The definition of the spectral gap for Hermitian systems can not be extended to non-Hermitian ones, as the eigenvalues are distributed over the complex plane. For non-Hermitian systems, one may introduce two different types of complex-energy gaps, the so-called point gap and line gap \cite{kawabata2019}. 
A non-Hermitian system is said to attain a point gap if its dispersion curve forms a loop that encircles a reference point $\omega_B\in\bbC$ and the crossing the base point defines a gap closing transition. A line gap, on the other hand, is defined by a line over the complex plane that does not intersect with the spectral bands of the system. We refer to Figure \ref{fig:gap_non-Hermtitian} for an illustration of these two types of spectral gaps. 

\begin{figure}[!htbp]
    \centering
    \includegraphics[height=5cm]{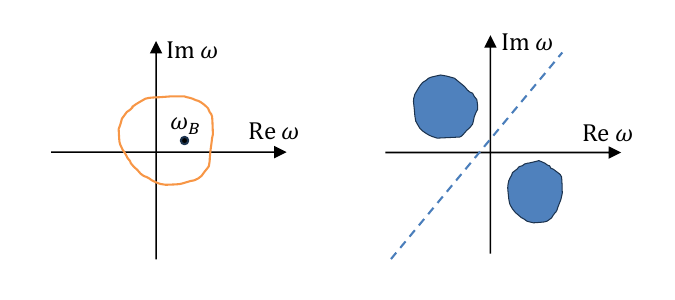}
    \vspace*{-15pt}
    \caption{Point gap (left) and line gap (right) for the spectrum of non-Hermitian operators. }
    \label{fig:gap_non-Hermtitian}.
\end{figure}

Similar to Hermitian photonic crystals, we distinguish two types of non-Hermitian photonic crystals:
\begin{itemize}
    \item [(i)] The spectral reciprocity holds such that $\omega_n(k)=\omega_n(-k)$ for all $n$.
    \item [(ii)] The spectral reciprocity is broken  with $\omega_n(k)\neq\omega_n(-k)$ for some $n$. 
\end{itemize}
Again the spectral reciprocity holds when the symmetry group of the period medium attain an operation $g$ such that $g \bPhi(k,\cdot) =  \bPhi(-k,\cdot)$. To break spectral reciprocity, one needs to exclude such operations in the symmetry group of the periodic media, such as the inversion symmetry and time-reversal symmetry.

A point gap occurs for a non-Hermitian photonic crystal when the spectrum is non-reciprocal, since the spectral symmetry $\omega_n(k)=\omega_n(-k)$ implies that the spectral band $\omega_n(k)$ typically forms a trivial gap with an empty interior region as illustrated in Figure \ref{fig:point_gap} (left).
The spectral non-reciprocity allows the two segments of the dispersion curve, $\{ \omega_n(k); \; k\in [-\pi,0]\}$ and $\{ \omega_n(k); \; k\in [0,\pi] \}$, to not overlap. Furthermore, when
\begin{equation}\label{eq:dis_periodicity}
 \omega_n(-\pi)=\omega_n(\pi),   
\end{equation}
the whole dispersion curve $\gamma_n:=\{ \omega_n(k); \; k\in \cB \}$ forms a closed loop (cf. Figure \ref{fig:point_gap}, right).

To demonstrate the existence of point gaps for non-reciprocal and non-Hermitian photonic crystals, let us revisit the three-layer periodic medium with the permittivity and permeability tensors in the form of \eqref{eq:para_three_layers}. We set all the parameters the same as in Example 3, except for perturbing $\eps_0$ from $13$ to $13+5\i$ for $\beps_A$ such that the system become non-Hermitian. The corresponding band structure is shown Figure \ref{fig:band_three_layer} (Right). It is seen that each of the first five bands forms a closed loop over the complex plane enclosing a non-trivial interior region. As such, each dispersion curve encircles a reference base point $\omega_B$ and the non-Hermitian photonic crystal attains point gaps. \\

\begin{figure}[!htbp]
    \centering
    \includegraphics[height=5cm]{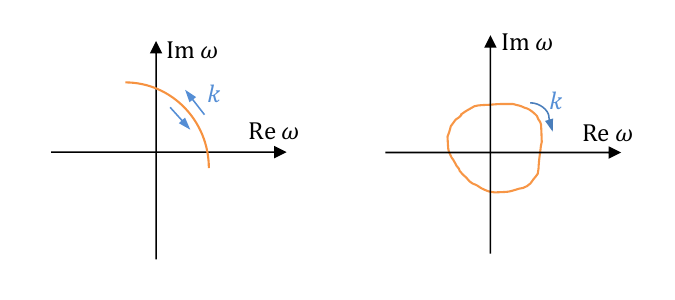}
    \vspace*{-15pt}
    \caption{Spectrum of non-Hermitian operators over the complex plane: a trivial gap when $\omega_n(k)=\omega_n(-k)$ (left) and a non-trivial gap when the spectral symmetry is broken (right).}
    \label{fig:point_gap}.
\end{figure}

\begin{remark}    
In general, the spectral non-reciprocity for one-dimensional photonic crystals can only be achieved when the medium parameter $\beps(z)$ or $\bmu(z)$ is anisotropic with nonzero off-diagonal entries, no matter the system is Hermitian or non-Hermitian. This is because the model for the electromagnetic wave propagation in a one-dimensional isotropic medium is a scalar second-order ODE for the electric or magnetic field.  The corresponding dispersion relation is determined by the nonlinear equation $\mbox{tr}(\cM(\omega)) =  2\cos(k)$ (see, for instance, \cite{lin2022}), which implies that $\omega(k)=\omega(-k)$.
\end{remark}

\medskip

\begin{remark}    
When considering the point gap in this work, we assume throughout that the continuous spectral band $\omega_n(k)$ satisfies $\omega_n(-\pi)=\omega_n(\pi)$ so that the dispersion curve $\gamma_n:=\{\omega_n(k); k\in\mathcal{B}\}$ itself forms a closed loop in the complex plane. We don't consider spectral braiding, where two dispersion curves $\gamma_n$ and $\gamma_m$ together form a closed loop in the complex plane \cite{zhang2023}.
\end{remark}

\section{Spectral topology and skin modes in non-Hermitian photonic crystals}\label{sec:non_Hermitian}
In this section, we investigate the skin modes in semi-infinite photonic crystals for which the dispersion curves for the quasi-periodic problem \eqref{eq:eig_prob1} attain point gaps.
Based on the discussions in the previous section, we consider
non-Hermitian and non-reciprocal photonic crystals such that the spectrum lies on the complex plane and the dispersion curves enclose nontrivial interior regions.
We show that the eigenfrequencies for the skin modes are closely related to the winding directions of the dispersion curves. In addition, the winding number of each dispersion curve can be characterized by a new topological invariant defined in terms of the eigenvalues of the underlying monodromy matrix.

\subsection{Topological index for the spectral gap}\label{sec:top_idx_spec_gap}
Consider a photonic crystal with the medium parameter $(\boldsymbol{\varepsilon}(z), \boldsymbol{\mu}(z))$. Let $\mathbb{T} := \{ e^{\i k}; k \in (-\pi,\pi] \}$ be the Brillouin zone over the complex plane.
For convenience of notation, for each $\xi\in\bbT$, we rewrite the eigenvalue problem \eqref{eq:eig_prob1} as
\begin{subequations}
\label{eq:eig_prob_complex}
\begin{align}
\label{eq:ODE_model_complex}
&  \frac{d}{dz} \bPhi(\xi; z) = \omega(\xi) \cQ^{-1}\cA(z) \bPhi(\xi; z),  \\ 
\label{eq:quasi_bnd_comlex}
& \bPhi(\xi; z+1) = \xi \bPhi(\xi; z).
\end{align}
\end{subequations}
The union of dispersion curves and their complement are denoted by
\[
\Gamma = \bigcup_{n=1}^{\infty} \gamma_{n}
\quad\mbox{and}\quad
\Gamma^c=\mathbb{C}\backslash \Gamma.
\]
If the dispersion curve $\gamma_n$ is a closed loop in the complex plane,
the corresponding dispersion relation may be viewed as a function from $\mathbb{T}$ to
$\bbC$, which, with the abuse of notation, is still denoted by $\omega_n$. In this scenario, we may express the dispersion curve as
\[
\gamma_{n}= \omega_{n}(\mathbb{T}).
\]

For any $\omega\in\mathbb{C}$, let $\lambda_{j}(\omega)$ ($j=1, 2, 3, ,4$) be the eigenvalue branches of the monodromy matrix $\cM(\omega)$ for the periodic medium that are defined in \eqref{eq:root_quartic}. We introduce an index at $\omega$ as
\[
 \Ind(\omega)= \frac{1}{2} \Big(\#\{ i: |\lambda_{i}(\omega)|<1\} - \#\{ i: |\lambda_{i}(\omega)|>1\}\Big).
\]

\begin{remark}\label{rmk:ind}
 In view of Proposition \ref{prop:rel_lambda_omega},  if $\omega\in\Gamma^c$, there holds $|\lambda_{j}(\omega)|\neq 1$ for $j=1, 2, 3, 4$. In addition, since $\prod_{j=1}^{4}\lambda_j(\omega)=1$, at most three eigenvalues lie inside or outside the unit circle. Therefore, $\Ind(\omega)$ can only take the value $0$ or $\pm 1$ for $\omega\in\Gamma^c$.
\end{remark}

\bigskip

The union of dispersion curves $\Gamma$ divides the complex plane into connected components. Let $D$ be one connected component in $\Gamma^c$, we first show that $\Ind(\omega)$ is invariant for $\omega\in D$. Before we prove the statement, we need the following result regarding the perturbation of the eigenvalues for matrices.

Let \( X = \{x_1, x_2, \dots, x_m\} \) and \( Y = \{y_1, y_2, \dots, y_m\} \) be two unordered \(m\)-tuples of complex numbers. 
The optimal matching distance between \( X \) and \( Y \) is defined as
\begin{equation}\label{eq:def_dxy}
    d(X, Y) = \min_{\pi \in S_m} \max_{1 \leq i \leq m} |x_i - y_{\pi(i)}|,
\end{equation}
where \( S_m \) is the symmetric group of degree \( m \), representing all permutations of $\{1, 2, \cdots, m\}$.
\begin{lemma}(\cite{elsner1985}) \label{lem:esner}
    For $A, B \in \bbC^{m\times m}$ with the eigenvalues $\Lambda_A:=\big\{ \lambda_1^A, \lambda_2^A, \cdots, \lambda_m^A\big \}$ and $\Lambda_B:=\big\{ \lambda_1^B, \lambda_2^B, \cdots, \lambda_m^B \big\}$ respectively, there holds
    \[
    d(\Lambda_A, \Lambda_B) \le m \, C^{1-1/m} \| A - B \|^{1/m}, \quad C = 2 \max\{ \|A\|, \|B\| \}.
    \]
\end{lemma}

\begin{lemma}\label{lem-invariance}
Let $D$ be a connected component of $\Gamma^c$.  The value of $\Ind(\omega)$ is independent of the choice of $\omega\in D$. 
\end{lemma} 

\begin{proof}
    
We only need to show that index $\Ind(\omega)$ is locally invariant for $\omega\in D$. Take an arbitrary $\omega_0 \in D$.

\medskip

\noindent\textbf{Case 1}. $\omega_0$ is not a branch point of the four eigenvalue functions $\lambda_{j}(\omega)$ ($j=1, 2, 3, 4$). Then the values $\lambda_{j}(\omega_0)$ are pairwise distinct. Note that each eigenvalue $\lambda_{j}(\omega_0)$ has modulus either strictly less than $1$ or strictly greater than $1$.
By the standard perturbation theory \cite{kato1966}, for $\omega$ sufficiently close to $\omega_0$, the four eigenvalue branches $\lambda_{j}(\omega)$ can be chosen continuously. Hence the quantities $|\lambda_{j}(\omega)|$ vary continuously in the small neighborhood of $\omega$.
As such, since $|\lambda_j(\omega_0)|\neq 1$, the function
$|\lambda_{j}(\omega)| - 1$
has constant sign for all $\omega$ sufficiently close to $\omega_0$. Consequently, the counts
$\#\{ i : |\lambda_{i}(\omega)| < 1 \}$ and $\#\{ i : |\lambda_{i}(\omega)| > 1 \}$ is locally constant near $\omega_0$, so is the index function $\Ind(\omega)$.  \\

\noindent\textbf{Case 2}. $\omega_0$ is a branch point of the four eigenvalue functions $\lambda_j(\omega)$. Then the monodromy matrix $\mathcal{M}(\omega)$ has eigenvalues with algebraic multiplicity strictly greater than one at $\omega_0$. 
Although the individual eigenvalue branch need not admit continuous parameterizations near $\omega_0$, by virtue of Lemma \ref{lem:esner},
the eigenvalue set $\{ \lambda_1(\omega), \lambda_2(\omega), \lambda_3(\omega), \lambda_4(\omega) \} $  depends continuously on $\omega$ with respect to the optimal matching metric $d$.

Recall that for $\omega_0\in D\subset\Gamma^c$, the entire eigenvalue multiset stays at a strictly positive distance from $\bbT$ (cf. Remark \ref{rmk:ind}) and $D$ is a connected component of $\Gamma^c$. We concludes that in a neighborhood of $\omega_0$, no eigenvalue crosses $\bbT$. 
As such
$\#\{ i : |\lambda_i(\omega)| < 1 \}$ and $\#\{ i : |\lambda_i(\omega)| > 1 \}$
are locally constant in $\omega$. This proves the index $\Ind(\omega)$ is locally invariant.

\end{proof}

In view of Lemma \ref{lem-invariance}, for a given connected component $D\subset\Gamma^c$, we define a component index associated with $D$ by
\[
 \Ind(D)=\Ind(\omega), \quad \omega\in D.
\]
It is clear that $\Ind(D)$ is well-defined.

\begin{definition}
 We call that a connected component $D$ of $\Gamma^c$ is a \textbf{principal gapped component} if 
 \[
   \Ind(D)=0. 
\] 
\end{definition}

\subsection{Homotopy invariance of $\Ind(D)$ for the principal gapped component}\label{sec:homotopy}

We view the non-Hermitian photonic crystal $(\boldsymbol{\varepsilon}(z), \boldsymbol{\mu}(z))$ as a continuous deformation from a Hermitian photonic crystal $(\boldsymbol{\varepsilon}_H(z), \boldsymbol{\mu}_H(z))$ whose spectrum is real. More precisely, we introduce two continuous functions
\[
h_{\varepsilon}, h_{\mu} :  \mathbb{R} \times [0,1] \to \mathbb{C}^{2\times 2},
\]
such that for all $z \in \mathbb{R}$,
\[
h_{\varepsilon}(z;0) = \boldsymbol{\varepsilon}_H(z), \qquad
h_{\mu}(z;0) = \boldsymbol{\mu}_H(z), \qquad
h_{\varepsilon}(z;1) = \boldsymbol{\varepsilon}(z), \qquad
h_{\mu}(z;1) = \boldsymbol{\mu}(z).
\]
Throughout, we require that $( h_\eps(\cdot;t), h_\mu(\cdot;t) )$ remains within the same coefficient class. More precisely, for each $t\in[0,1]$, the tensors $h_\eps(\cdot;t)$ and $h_\mu(\cdot;t)$ are $1$-periodic in $z$, bounded, piecewise continuous, retain the transverse block structure, and depend continuously on $t$ in the $L^\infty$ norm. This guarantees that the Floquet--Bloch transform and the monodromy matrix $\cM(\omega;t)$ is well defined for each $t$. Moreover, since the transfer matrix $\cT_{\omega}(z, 0)$ solves the Volterra equation $$\cT_{\omega}(z, 0)=\cI+\int_0^z\omega\cQ^{-1}\cA(s;t)\cT_{\omega}(s,0)\,ds.$$
A Gronwall estimate shows that $t\mapsto\cM(\omega;t)$ is continuous in $(\omega, t)$ over a compact set. Consequently the eigenvalue multiset $\{\lambda_j(\omega;t)\}$ is jointly continuous in $(\omega,t)$ with respect to the optimal matching metric \eqref{eq:def_dxy}.

For each $t \in [0,1]$, the $t$-dependent band diagram and its complement are given by
\begin{equation}\label{eq:Gamma_t}
    \Gamma_t = \bigcup_{n=1}^{\infty} \gamma_{n,t}
\quad\mbox{and}\quad
\Gamma_t^c=\mathbb{C}\backslash \Gamma_t.
\end{equation}

Let $\lambda_{j}(\omega; t)$ ($j=1, 2, 3, 4$) be the eigenvalue branches of the monodromy matrix $\cM(\omega;t)$ for the periodic medium $(h_\eps(z;t), h_\mu(z;t))$ that are defined in \eqref{eq:root_quartic}. For each connected component $D_t\subset \Gamma_t^c$, we denote the t-dependent index component by
\[
 \Ind(D_t; t):=\frac{1}{2} \Big(\#\{ i: |\lambda_{i}(\omega; t)|<1\} - \#\{ i: |\lambda_{i}(\omega; t)|>1\}\Big), \quad \omega\in D_t.
\]

We first show that $\Gamma_0^c$ is a principal gapped component for the Hermitian medium $(\beps_H(z), \bmu_H(z))$ under the following assumption. 

\begin{assumption}\label{assum1}
The spectrum of the Hermitian photonic crystal $(\beps_H(z), \bmu_H(z))$ contains a band gap. As a result, the set $\Gamma_0^c$ is a connected and unbounded domain in the complex plane.
\end{assumption}

\begin{thm} \label{lem-0}
Under Assumption \ref{assum1}, $\Gamma_0^c$ is principal gapped component, i.e., 
$\Ind(\Gamma_0^c; 0)=0$. 
\end{thm}

\begin{proof}
By Assumption \ref{assum1}, there exists a real frequency $\omega_*$ lying in a spectral gap of the Hermitian system. Since $\Gamma_0^c$ is connected and the index is constant on each connected component by Lemma \ref{lem-invariance}, it suffices to compute $\Ind(\omega_*;0)$.

From Proposition \ref{prop:rel_lambda_omega}, no eigenvalue of $\cM(\omega_*;0)$ lies on the unit circle $\bbT$, i.e., $|\lambda_j(\omega_*;0)|\neq 1$ for $j=1, 2, 3, 4$. Furthermore, by virtue of Lemma \ref{lem:flux}, the eigenvalues occur in reciprocal-conjugate pairs $\{\lambda,1/\bar\lambda\}$. Namely, one of this pair lies inside $\bbT$ and the other lies outside $\bbT$. 
Hence, exactly two eigenvalues of $\cM(\omega_*;0)$ satisfy $|\lambda(\omega_*;0)|<1$ and exactly two satisfy $|\lambda(\omega_*;0)|>1$. Therefore, we conclude that $\Ind(\Gamma_0^c;0)=0$.

\end{proof}

\begin{thm} \label{thm:ind_homotopy}
Consider the periodic media $\big(h_\eps(\cdot,t), h_\mu(\cdot,t)\big)$ with $\big(h_\eps(\cdot,0), h_\mu(\cdot,0)\big)=(\beps_H(z), \bmu_H(z))$ and $\big(h_\eps(\cdot,1), h_\mu(\cdot,1)\big)=(\beps(z), \bmu(z))$. Let $\Gamma_t$ and $\Gamma_t^c$ be the union of the corresponding dispersion curves and their complement.
Given a connected component $D\subset\Gamma_1^c$. If Assumption \ref{assum1} holds for $\big(h_\eps(\cdot,0), h_\mu(\cdot,0)\big)$, and there exists a continuous curve $g: [0, 1] \to \mathbb{C}$ such that 
\begin{itemize}
\item[(i)] 
For each $0\leq t \leq 1$, $g(t) \in \Gamma_t^c $.

\item[(ii)]
$g(1) \in D \subset\Gamma_1^c$;
\end{itemize}
Then $D$ is a principal gapped component, i.e. 
$Ind(D; 1)=0$. 
  
\end{thm}

\begin{proof}

For each $0\leq t \leq 1$, we define a component index function
\[
G(t)=\frac{1}{2} \Big(\#\{ i: |\lambda_{i}(g(t);t)|<1\} - \#\{ i: |\lambda_{i}(g(t); t)|>1\}\Big). 
\]
By Theorem \ref{lem-0}, we have $G(0)=0$; thus, to show that $G$ vanishes identically, it suffices to establish that it is locally constant on $[0, 1]$. 
To this end, we fix $t_0 \in [0,1]$ and show that $G(t)$ is constant for $t$ sufficiently close to $t_0$.

First, since $g(t)\in\Gamma_t^c$, in view of Proposition \ref{prop:rel_lambda_omega}, $|\lambda_{j}(g(t_0); t_0)|<1$ or $|\lambda_{j}(g(t_0); t_0)|>1$ holds for any $j\in\{1, 2, 3, 4\}$.
Using the perturbation theory in Lemma \ref{lem:esner} for the monodromy matrix $\mathcal{M}(g(t);t)$ and the continuous deformation of the medium parameters $(h_\eps(\cdot;t), h_\mu(\cdot;t))$, the eigenvalue set $\{\lambda_{j}(g(t);t): j=1, 2, 3, 4\}$ is continuously in $t$ with respect to the optimal matching metric $d$ when $t$ is sufficiently close to $t_0$.
In particular, when $g(t_0)$ is not a branch point of $\{\lambda_{j}(g(t);t)$, each of the four eigenvalue branches $\lambda_{j}(g(t);t)$ is continuous for $t$ near $t_0$.  Therefore, the counts
$\#\{ i : |\lambda_{i}(g(t);t)| < 1 \}$ and $\#\{ i : |\lambda_{i}(g(t);t)| > 1 \}$
are locally constant in $t$, so is $G(t)$. This completes the proof of the lemma.

\end{proof}

\begin{remark}
The function $\omega:=g(t)$ in the above theorem describes the evolution of the frequency $\omega$ on the complex plane, starting from the initial value $g(0)=\omega_*$. Assumptions~(i) and~(ii) require that, throughout the continuous deformation of the medium parameters $(h_\varepsilon(\cdot;t), h_\mu(\cdot;t))$ for $t\in[0,1]$, the trajectory $\omega=g(t)$ remains in the complement of the dispersion curves $\Gamma_t$.
\end{remark}

\begin{figure}[!htbp]
    \centering
    \vspace*{-15pt}
    \includegraphics[height=5cm]{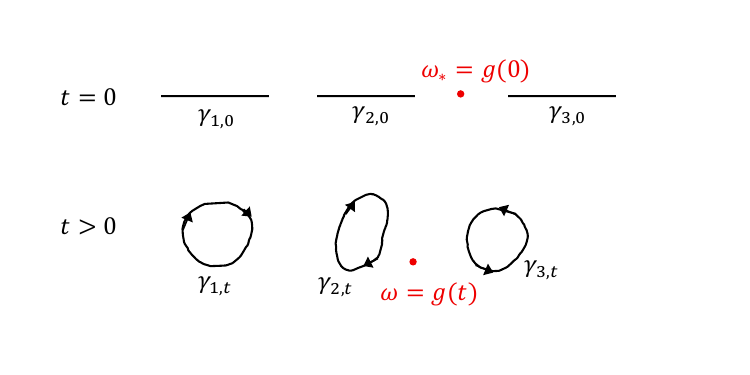}
    \vspace*{-25pt}
    \caption{The evolution of the frequency $\omega:=g(t)$ on the complex plane, starting from the initial value $g(0)=\omega_*$. }
    \label{fig:gt}.
\end{figure}

\subsection{Jump of $\Ind(D)$ across the dispersion curve}\label{sec:jump_ind}

The dispersion curves $\Gamma_1$ for the periodic medium $(\boldsymbol{\varepsilon}(z), \boldsymbol{\mu}(z))$ divide the whole complex plane into connected components (cf. Figure \ref{fig:dispersion_curve}). Let us consider a closed dispersion curve $\gamma_{n}$, which is the image of the $n$-th continuous dispersion function $\omega_{n}: \mathbb{T} \to \mathbb{C}$.  $\gamma_n$ partitions the complex plane into a set of disjoint connected components. 
Among them, there exists one unique unbounded connected component $D_{n}^{e}$ and a finite number of bounded components $D_n^i$ ($i=1, 2, \cdots, i_0$). We assume that each \(D_n^i\) is a \textbf{Jordan domain} bounded by a \textbf{Jordan curve} \(\gamma_n^i \subset \gamma_n\) that is smooth everywhere except at finitely many corners. In addition, for $m\neq n$, $\gamma_m^i\neq\gamma_n^i$ as sets.
For a given smooth point $\zeta \in \gamma_n^i$, we define the \textbf{outward unit normal} $\nu(\zeta)$ as the one that points into the unbounded domain $D_n^e$. 


\begin{figure}[!htbp]
    \centering
    \vspace*{-15pt}
    \includegraphics[height=5cm]{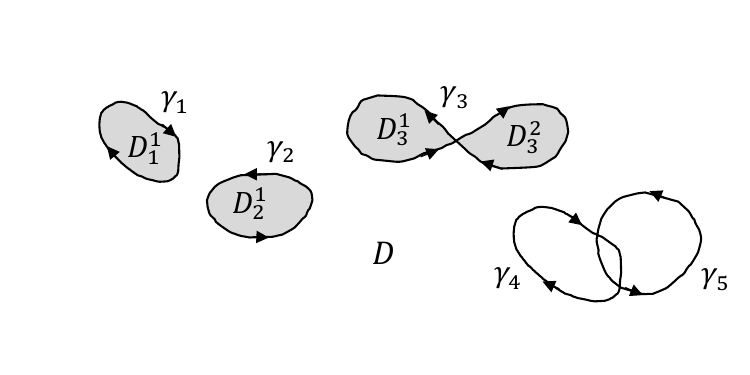}
    \vspace*{-25pt}
    \caption{A schematic plot of dispersion curves $\gamma_1$, $\gamma_2$, $\gamma_3$, $\cdots$ for the non-Hermitian photonic crystal. Each dispersion curve $\gamma_n$ forms a finite number of bounded components $D_n^i$ ($i=1, 2, \cdots, i_0$) over the complex plane. $D \subset \Gamma_1^c$ is a connected component in the complement of the dispersion curves.  }
    \label{fig:dispersion_curve}.
\end{figure}

For the simple closed curve $\gamma_n^i$, we say that it is \textbf{oriented counterclockwise} (or positively oriented) if the mapping $\omega_{n}$ preserves the counterclockwise orientation of the unit circle $\mathbb{T}$ for $\omega_n(\xi)\in \gamma_n^i$. More precisely,
let $\xi_0 = (\cos(k_0), \sin(k_0))\in \mathbb{T}$ be such that $\omega_{n}(\xi_0)$ is a smooth point of $\gamma_n^i$. Let $\tau=(-\sin(k_0), \cos(k_0))$ be the unit tangent vector to $\mathbb{T}$ at $\xi_0$, and
$\omega_{n}'(\xi_0)v$ be the tangent vector to the curve $\gamma_n^i$ at the point $\zeta_0 = \omega_{n}(\xi_0)$.
The curve $\gamma_n^i$ is oriented counterclockwise if the outward normal $\nu(\zeta_0)$ points in the direction of rotating the tangent vector $\omega_{n}'(\xi_0)v$ clockwise by $\pi/2$. Otherwise, the curve $\gamma_n^i$ is said to be oriented clockwise (or negatively oriented).

Let $D$ be a connected component in $\Gamma_1^c$. For instance, one may set $\displaystyle{D=\bigcap_{n=1}^\infty D_n^{e}}$ if the intersection of all $D_n^{e}$ is connected. Let the assumptions in Theorem \ref{thm:ind_homotopy} hold, namely, the periodic medium $(\boldsymbol{\varepsilon}(z), \boldsymbol{\mu}(z))$ can be deformed continuously from a Hermitian medium with a spectral gap and there exists a continuous function $g(t)$ such that the trajectory $\omega=g(t)$ remains in the complement of the dispersion curves $\Gamma_t$ and $g(1)\in D$. Then
\begin{equation}\label{eq:D}
    D \; \mbox{is a principal gapped component with} \; \Ind(D)=0. 
\end{equation}

In this subsection, we study the topological index of a bounded connected component $D_n^i$ that is adjacent to the principal gapped component $D$. Our main result is that the topological index attains a jump when crossing the dispersion curve $\gamma_n$ from $D$ to $D_n^i$. 
We first give the winding number of the dispersion curve $\gamma_n$.

\begin{thm} \label{thm-winding}
Assume that $\gamma_n=\omega_n(\bbT)$ is a closed dispersion curve and does not intersect with other dispersion curves $\gamma_m$ for $m\neq n$. Further assume that $\omega_n(\cdot)$ can be extended analytical near $\bbT$. 
Let $D_n^i$ be a Jordan domain partitioned by $\gamma_n$, with the boundary given by the simple closed curve $\gamma_n^i$.
If $\gamma_n^i$ is positively oriented, then the winding number 
\begin{equation*}
  W(\gamma_{n}; \omega_B) = \frac{1}{2\pi \i} \int_{\xi\in\bbT}   \frac{d\omega_n(\xi)}{\omega_n(\xi)-\omega_B}=1
\end{equation*}
for $\omega_B \in D_n^i$. 
If $\gamma_n^i$ is negatively oriented, then the winding number $W(\gamma_{n}, \omega_B)=-1$ for $\omega_B \in D_n^i$.
\end{thm}

\begin{proof}
This is the consequence of the standard argument principle together with the additivity of winding numbers and the orientation convention. We give the proof when $\gamma_n^i$ is positively oriented, and the proof is similar when $\gamma_n^i$ is negatively oriented. If $\omega_n(\xi)$ is injective for $\xi\in\bbT$,
then $\gamma_n$ is a simple closed curve on the complex plane and there holds $\gamma_n=\gamma_n^1$. By the Jordan Curve Theorem, $\gamma_n^1$ divides the complex plane into a bounded component $D_n^1$ and an unbounded component $\bbC\backslash \overline{D_n^1}$. Since the mapping $\omega_n(z)$ is orientation preserving, the winding number $W(\gamma_n; \omega_B)=W(\gamma_n^1; \omega_B)=1$ for any $\omega_B\in D_n^i$. 

If $\omega_n(\xi)$ is not injective for $\xi\in\bbT$,  then $\gamma_n=\omega_n(\bbT)$ can be decomposed as a finite union of Jordan curves $\gamma_n=\cup_{j=1}^{i_0}\gamma_n^j$ with each Jordan curve $\gamma_n^j$ enclosing a domain $D_n^j$. For any $\omega_B$ in $D_n^i$, it is clear that the winding number  $W(\gamma_n^i; \omega_B)=1$ and $W(\gamma_n^j; \omega_B)=0$ for $j\neq i$, since $\omega_B\in D_n^i$ lies in the exterior of the domain enclosed by $\gamma_n^j$. Therefore, $W(\gamma_{n}, \omega_B)= \sum_{j=1}^{i_0} W(\gamma_n^j; \omega_B)=1$. 

\end{proof}

Next, we investigate the topological index of the bounded component $D_n^i$.

\begin{thm} \label{thm:Ind_Dn}
 Let $D_n^i$ be a Jordan domain partitioned by $\gamma_n$ as in Theorem \eqref{thm-winding}, with the boundary given by the simple closed curve $\gamma_n^i$, and $D\subset \Gamma_1^c$ be a principle gapped component that satisfies \eqref{eq:D}.
Assume that
\begin{itemize}
    \item [(i)] There exists $\xi_0\in \bbT$ such that $\zeta_0=\omega_n(\xi_0)\in\gamma_n^i$ is a smooth point of $\gamma_n$ with a well-defined outward unit normal vector $\nu$. Additionally, $\zeta_0 + t \nu \in D$ and
$\zeta_0 - t \nu \in D_n^i$ for sufficiently small $t>0$.
    \item [(ii)] The function $\omega_n(\xi)$ defined via \eqref{eq:eig_prob_complex} is analytic at $\xi_0$. 
 
\end{itemize}
If $\gamma_n^i$ is positively oriented, then 
    \[
    \Ind(D_n^i; 1) =  W(\gamma_{n}; \omega_B) = 1 \quad \mbox{for} \; \omega_B \in D_n^i.
    \]
Otherwise, 
\[
    \Ind(D_n^i; 1) =  W(\gamma_{n}; \omega_B) = -1.
\]
\end{thm}

\begin{remark}
A sufficient condition for Assumption (ii) in the theorem above is that $(\zeta_0,\xi_0)$ is a simple zero of $P(\omega,\xi):=\det(\cM(\omega)-\xi I)$. The corresponding eigenvalue is simple and separated from the others. The analyticity of $\omega_n(\xi)$ then follows from the analytic implicit function theorem and the perturbation theory \cite{kato1966}.
\end{remark}

\begin{proof}

We assume that $\gamma_n^i$ is positively oriented. \\
\emph{Step 1}. Since $\omega_n(\xi)$ is analytic at $\xi_0$, we may expand $\omega_n(\xi)$ as
\begin{equation*}
    \omega_n(\xi) = \omega_n(\xi_0) + \sum_{m=1}^\infty \alpha_n (\xi-\xi_0)^m.
\end{equation*}
Note that $\zeta_0=\omega_n(\xi_0)$ is a smooth point of $\gamma_n$, we deduce that $\omega_n'(\xi_0) \neq 0$. As such $\omega_n(\xi)$ is a conformal mapping at $\xi_0$ and $\omega_n(\xi)$ is locally injective at $\xi_0$.

Let $\tau=(-\sin(k_0), \cos(k_0))$ and $\nu=(\cos(k_0), \sin(k_0))$ be the unit tangent and normal vectors to $\mathbb{T}$ at $\xi_0$ respectively. Then
 at $\zeta_0=\omega_n(\xi_0)\in\gamma_n^i$, the normal vector $\omega_n'(\xi_0) \cdot \nu$ to the curve $\gamma_n$ is obtained by rotating the tangent vector $\omega_n'(\xi_0) \cdot \tau$ clockwise by $\pi/2$.  Since $\gamma_n^i$ is positively oriented, or the mapping $\omega_n(\xi)$ preserves the orientation on $\mathbb{T}$ for $z$ near $\xi_0$, we deduce that the normal vector $\omega_n'(\xi_0) \cdot \nu$ 
 points into the unbounded exterior domain $\bbC\backslash\overline{D_n^i}$. In view of the assumption (i), this implies that for sufficiently small $t>0$,
 \begin{equation}\label{eq:omega_n_mapping}
\omega_n(\xi_0 + t\nu) \in D \quad\mbox{and} \quad
\omega_n(\xi_0 - t\nu) \in D_n^i
 \end{equation}
respectively.

\emph{Step 2. Computing $Ind(D_n^i; 1)$.} Recall that $\omega_n(z)$ is a conformal mapping at $\xi_0$. There exists an inverse map near $\zeta_0$, which we denote by $f: B(\zeta_0, \delta_0) \to B(\xi_0, \eta_0)$. By virtue of \eqref{eq:omega_n_mapping}, we have
\begin{equation}\label{eq:f_mapping1}
f\big(B(\zeta_0, \delta_0)\cap D\big) \subset B(\xi_0, \eta_0) \cap \{\xi: |\xi|>1\} 
\end{equation}
and
\begin{equation}\label{eq:f_mapping2}
f\big(B(\zeta_0, \delta_0)\cap D_n^i\big) \subset B(\xi_0, \eta_0) \cap \{\xi: |\xi|<1\}.
\end{equation}
Since $\Ind(D)=0$ implies that there are exactly two eigenvalues of the monodromy matrix $\cM(\omega)$, say $\lambda_{1}(\omega)$ and  $\lambda_{2}(\omega)$, that attain modulus strictly greater than $1$ for $\omega\in D$. From  \eqref{eq:f_mapping1}, we see that the mapping $f$ coincides with one of these two eigenvalue functions. Without loss of generality, let us assume that $f(\omega)=\lambda_1(\omega)$ for $\omega\in B(\zeta_0, \delta_0)$. In view of \eqref{eq:f_mapping2}, there holds $|\lambda_1(\omega)|<1$ for $\omega\in B(\zeta_0, \delta_0)\cap D_n^i$. On the other hand, due to the uniqueness of the inverse map $f$ in $B(\zeta_0, \delta_0)$, we deduce that
\begin{equation}
   |\lambda_2(\omega)|>1, |\lambda_3(\omega)|<1, |\lambda_4(\omega)|<1
\end{equation}
for $\omega\in B(\zeta_0, \delta_0)$. We conclude that 
\[
Ind(D_n^i; 1) =1. 
\]

Following the parallel lines, it can be shown that $Ind(D_n^i; 1) =-1$  when $\gamma_n^i$ is negatively oriented. 
\end{proof}

The above theorem gives the topological index of a Jordan domain $D_n^i$, which is enclosed by a single dispersion curve $\gamma_n$ and does not intersect with other dispersion curves. This result extends to the more general setting in which the dispersion curves are allowed to intersect, as stated in the following corollary. 

\begin{coro} \label{prop:Ind_Dn}
Assume that $\gamma_n=\omega_n(\bbT)$ is a closed dispersion curve, and that $\omega_n(\cdot)$ can be extended analytically near $\bbT$. Let $D_n^i$ be a Jordan domain partitioned by $\gamma_n$ with the simple closed curve $\gamma_n^i\subset \gamma_n$ as its boundary. Assume that $\gamma_n\cap\gamma_m\neq\emptyset$ for some $m\neq n$. Let
\[
\widetilde D \subset D_n^i\backslash \bigcup_{m\neq n}D_m
\]
be a connected component adjacent to the principal gapped component $D$ through a regular point of $\gamma_n^i$ satisfying the assumptions (i)(ii) in Theorem \ref{thm:Ind_Dn}. Then
    \[
    \Ind\left(\widetilde D; 1\right) =  W(\gamma_{n}; \omega_B) = 1 \quad \mbox{for} \; \omega_B \in \widetilde D
    \]
when $\gamma_n^i$ is positively oriented, and
    \[
    \Ind\left(\widetilde D; 1\right) =  W(\gamma_{n}; \omega_B) = -1 \quad \mbox{for} \; \omega_B \in \widetilde D
    \]
when $\gamma_n^i$ is negatively oriented.    
In the above, $\displaystyle{D_m=\bigcup_{i=1}^{i_0^m} D_{m}^i}$ is the set of all bounded domains partitioned by the dispersive curve $\gamma_m$.    
\end{coro}

The proof for the corollary is similar to the proof of Theorem \ref{thm:Ind_Dn}.
It is clear that \eqref{eq:f_mapping1} still holds, and \eqref{eq:f_mapping2} holds with $D_n^i$ replaced by $\widetilde D$. As such, one can still conclude that only three eigenvalues attain modulus smaller than $1$ in the component $\widetilde D$ when $\gamma_n^i$ is positively oriented.

\bigskip

\begin{remark}
In the one-dimensional setting considered in this work, the domain partitioned by $\gamma_n$ cannot have a winding number greater than or equal to two. To illustrate this, let us consider the positively oriented dispersion curve $\gamma_n$ shown in Figure~\ref{fig:winding2}, which partitions the complex plane into the three regions $D_n^1$, $D_n^2$, and $D_n^e$. The component $D_n^2$ is a higher-winding component with
\[
 W(\gamma_n;\omega_B)=2 \quad \mbox{for} \; \omega_B\in D_n^2.
\]
An extension of the argument in Theorems~\ref{thm:Ind_Dn} by using the conformal mapping would show that $\Ind(D_n^2)=2$, which contradicts with the fact that $\Ind(D_n^2) \in \{0, \pm 1\}$ (cf. Remark \ref{rmk:ind}).
\end{remark}
   
\begin{figure}[!htbp]
    \centering
    \includegraphics[height=4.5cm]{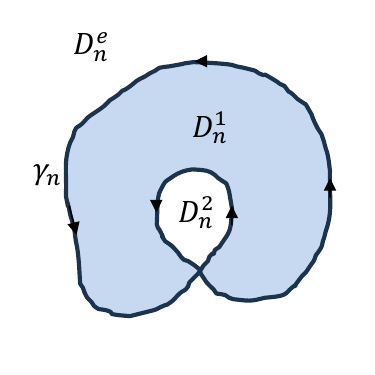}
    \vspace*{-25pt}
    \caption{The domains enclosed by $\gamma_n$ are not Jordan domains.}
    \label{fig:winding2}.
\end{figure}

\subsection{Skin modes in semi-infinite photonic crystals with non-trivial spectral topology}

We consider a semi-infinite photonic crystal sitting in the interval $I_-:=(-\infty, 0)$ or $I_+ :=(0, \infty)$.
For $z\in I_\pm$, the medium parameter is given by $(\boldsymbol{\varepsilon}(z), \boldsymbol{\mu}(z))$, and the ambient medium outside the photonic crystal is assumed to be vacuum.
The corresponding spectral problem is formulated as finding $(\omega, \bPhi) \in \bbC \times (L^2(I_\pm))^4$ such that
\begin{subequations}
\label{eq:semi_infinite}
\begin{align}
   & \frac{d}{dz} \bPhi(z) = \omega \cQ^{-1}\cA(z) \bPhi(z), \quad z \in I_\pm;  \label{eq:ODE_system_semi_infinite} \\
   & \bPhi(0)= \left(\begin{matrix}
        \bE_0(0) \\ \bH_0(0) \label{eq:IC_semi_infinite}
    \end{matrix} \right).
\end{align}    
\end{subequations}
In the above, $\{\bE_0(z), \bH_0(z) \}$ is the outgoing electromagnetic wave in the ambient medium.  More explicitly,  $\{\bE_0(z), \bH_0(z) \}$ can be expressed as
\begin{equation*}
   \left(\begin{matrix}
        \bE_0(z) \\ \bH_0(z)
    \end{matrix} \right)
    =    \left(\begin{matrix}
        \bc_1 \\ \bc_2
    \end{matrix} \right)e^{\i k_0 z}
    \quad \mbox{and} \quad
    \left(\begin{matrix}
        \bE_0(z) \\ \bH_0(z)
    \end{matrix} \right)
    =    \left(\begin{matrix}
        \bc_1 \\ \bc_2
    \end{matrix} \right)e^{-\i k_0 z}
\end{equation*}
for $z>0$ and $z<0$ respectively. In the above, $k_0:=\omega/c$ is the wavenumber in the vacuum, and $\bc_1$, $\bc_2 \in \bbC^2$ are constant vectors that satisfy 
\begin{equation}\label{eq:c1c2}
    \bc_2 = Q \bc_1    \quad \mbox{and} \quad \bc_2 = -Q \bc_1 
\end{equation}
when $z>0$ and $z<0$ respectively.

The corresponding infinite periodic problem is given by \eqref{eq:eig_prob_complex}, which attains the 
dispersion curves $\{\gamma_{n} \}_{n\in \bbN^+}$ over the complex plane.
We consider $\omega \in \bbC$ that is located in one of bounded domains partitioned by 
the dispersion curve $\gamma_{n}$. For clarity of presentation, let us assume that $\gamma_n\cap\gamma_m=\emptyset$ for $m\neq n$, and consider $\omega\in D_n^i$, where $D_n^i$ is a Jordan domain for which Theorem \ref{thm:Ind_Dn} applies. Then the eigenvalues of the semi-infinite problem \eqref{eq:semi_infinite} are described in the following theorem. 
Note that the assumption on the non-intersection of the dispersion curves is not essential as long as we restrict $\omega$ to a smaller domain $D_n^i\backslash \bigcup_{m \neq n} D_m$; we refer to Corollary \ref{prop:Ind_Dn} and Remark \ref{rmk:edge_modes_intersect} for the discussions of this scenario.
 
\begin{thm}\label{thm:edge_modes}
Let $D_n^i$ be a Jordan domain partitioned by $\gamma_n$ as in Theorem \eqref{thm-winding}. 
    \begin{itemize}
        \item [(i)]  If  
        \begin{equation}\label{eq:wind_p}
            \Ind(D_n^i; 1) =  W(\gamma_{n}; \omega_B) = 1 \quad \mbox{for} \; \omega_B \in D_n^i,
        \end{equation}
         then each $\omega \in D_n^i$ is an eigenvalue for the semi-infinite problem \eqref{eq:semi_infinite} defined over the interval $I_+$, as long as $\omega$ is not a branch point of $\lambda_j(\omega)$ (j=1,2,3,4) for the monodromy matrix $\cM(\omega)$.
         The corresponding eigenfunction $\bPhi(z)$ decays exponentially as $z\to\infty$.
        \item [(ii)]  If 
        \begin{equation}\label{eq:wind_n}
            \Ind(D_n^i; 1) =  W(\gamma_{n}; \omega_B) = -1 \quad \mbox{for} \; \omega_B \in D_n^i,
        \end{equation}
    then each $\omega \in D_n^i$ is an eigenvalue for the semi-infinite problem \eqref{eq:semi_infinite} defined over the interval $I_-$, as long as $\omega$ is not a branch point of $\lambda_j(\omega)$ (j=1,2,3,4) for the monodromy matrix $\cM(\omega)$.
         The corresponding eigenfunction $\bPhi(z)$ decays exponentially as $z\to-\infty$.
    \end{itemize}
\end{thm}

\bigskip
\begin{remark}
By virtue of Theorem \ref{thm-winding}, under suitable assumptions,  \eqref{eq:wind_p} and \eqref{eq:wind_n} hold when $\gamma_n^i$, the boundary of the $D_n^i$, is positively and negatively oriented respectively. Hence, the above theorem relates the existence of the skin modes for the semi-infinite problem \eqref{eq:semi_infinite} with the winding direction of $\gamma_n^i$.
\end{remark}

\begin{proof}

If $ \Ind(D_n^i; 1)=1$, there exist three eigenvalues of the monodromy matrix $\cM(\omega)$ with modulus smaller than 1, and one eigenvalue with modulus larger than 1. Furthermore, $\omega$ is not a branch points of $\lambda_j(\omega)$ (j=1,2,3,4). Therefore, without loss of generality, we have
\begin{equation*}
    |\lambda_{1}(\omega)| < |\lambda_{2}(\omega)| < |\lambda_{3}(\omega)| < 1 < |\lambda_{4}(\omega)|.
\end{equation*}  
Let $\bPhi_j:=[\bE_j, \bH_j]^T$ be the corresponding eigenvectors of $\cM(\omega)$ for $j=1,2,3,4$.
We choose the initial value of $\bPhi(z)$ at $z=0$ as the linear combination:
\begin{equation}\label{eq:phi0}
    \bPhi(0) = \alpha_1 \bPhi_1 + \alpha_2 \bPhi_2 + \alpha_3 \bPhi_3.
\end{equation}
Then the solution of the ODE system \eqref{eq:ODE_system_semi_infinite} with the above initial value takes the following form at $z=n$:
\begin{equation}\label{eq:phi}
    \bPhi(z) = \alpha_1 \lambda_{1}^n \bPhi_1 + \alpha_2 \lambda_{2}^n \bPhi_2 + \alpha_3 \lambda_{3}^n \bPhi_3,
\end{equation}
which decays exponentially as $n\to\infty$. We only need to show that there exists $\alpha_1$, $\alpha_2$, and $\alpha_3$ such that the initial condition \eqref{eq:IC_semi_infinite} is satisfied.  In view of \eqref{eq:c1c2}, this boils down to solving the following linear system of two equations:
\begin{equation}\label{eq:linear_sys}
  \bB  \balpha = \bf{0}, \quad\mbox{where} \; \bB= [\bH_1,\bH_2,\bH_3] + Q \, [\bE_1,\bE_2,\bE_3] \in \bbC^{2\times3}.
\end{equation}
The above linear system attains at least one solution $\balpha:= [\alpha_1, \alpha_2, \alpha_3]^T \in \bbC^3$. This proves that  $\omega$ is an eigenvalue of the problem \eqref{eq:semi_infinite}, with the corresponding eigenfunction decaying in the form of \eqref{eq:phi}.

The argument is similar when $\Ind(D_n^i;1)=-1$. Now three eigenvalues of  $\cM(\omega)$ attain modulus larger than 1, and one eigenvalue has modulus smaller than 1. By ordering them as
$$ |\lambda_4(\omega)|<1<|\lambda_1(\omega)|\le|\lambda_2(\omega)|\le|\lambda_3(\omega)|, $$
we express the solution at $z=-n$ as $\bPhi(-n)=\sum_{j=1}^3\alpha_j\lambda_j^{-n}\bPhi_j$, which decays exponentially as $n\to\infty$. The continuity condition in \eqref{eq:c1c2} again yields a linear system $\bB\balpha=\mathbf{0}$ with $\bB=[\bH_1,\bH_2,\bH_3]-Q[\bE_1,\bE_2,\bE_3]\in\bbC^{2\times3}$, which admits a nontrivial solution.
This proves the theorem.

\end{proof}

\begin{remark}
    Theorem \ref{thm:edge_modes} holds if other bound boundary conditions are imposed. For instance, if the photonic crystal is terminated with a perfect electric or magnetic conducting wall at $z=0$, which gives
\begin{equation*}
    \bE(0) = \bf{0} \quad \mbox{and} \quad \bH(0) = \bf{0},
\end{equation*}
for the initial condition \eqref{eq:IC_semi_infinite}, respectively. Similarly, one can find the initial value in the form of \eqref{eq:phi0} that satisfies the above boundary condition.
\end{remark}

\bigskip

\begin{remark}\label{rmk:edge_modes_intersect}
If $\gamma_n\cap\gamma_m\neq\emptyset$ for some $m \neq n$, then Theorem \ref{thm:edge_modes} holds on any connected component $\widetilde D\subset D_n^i\backslash \bigcup_{m \neq n} D_m$ for which Corollary \ref{prop:Ind_Dn} applies. More specifically, $\omega$ is an eigenvalue of \eqref{eq:semi_infinite} over the interval $I_+$ and $I_-$ when
\[
   \Ind(\widetilde D;1)=1
   \quad\mbox{and}\quad
   \Ind(\widetilde D;1)=-1
\]
respectively.    
\end{remark}

\bigskip

We now state the bulk-edge correspondence, which relates the lower bound and the exact number of skin modes to the bulk topological invariant.

\begin{prop}  (\textbf{Bulk-edge correspondence})
Let $D_n^i$ be a Jordan domain partitioned by $\gamma_n$ as in Theorem \eqref{thm-winding}.  Set $s:=\Ind(D_n^i; 1)$ and assume that $s=\pm1$. For each non-branch-point $\omega\in D_n^i$ that is the eigenvalue for \eqref{eq:semi_infinite}, the dimension of the eigensapce is at least $1$, and this lower bound is independent of $\mathrm{rank}(\bB)$. If $\mathrm{rank}(\bB)=2$, then the dimension  of the eigenspace is exactly $1$.
\end{prop}

\begin{proof}
We assume that $s=1$. From the proof of Theorem \ref{thm:edge_modes}, by expressing the skin mode $\bPhi(z)$ as \eqref{eq:phi}, the spectral problem \eqref{eq:semi_infinite} reduces to solving the linear system \eqref{eq:linear_sys}. Since $\bB\in\bbC^{2\times3}$, there holds $\dim\ker\bB\ge1$, regardless of the rank of $\bB$. If $\mathrm{rank}(\bB)=2$, then $\dim\ker\bB=1$.

\end{proof}

\section*{Acknowledgments}
JL is partially supported by NSF grant DMS-2410645 and Auburn University Research Support Program. HZ is partially supported by Hong Kong RGC grant GRF 16304621 and NSFC grant 12371425.
JL also gratefully acknowledges the support and hospitality provided by Hong Kong University of Science and Technology during his visit and when part of this work was performed.

\begin{appendices}

\section{Roots of quartic functions}\label{sec:roots_quartic_fun}
Consider solving a quartic equation
\begin{equation}\label{eq:quartic}
    \lambda^4 + a_3 \lambda^3 + a_2 \lambda^2 + a_1 \lambda + a_0 = 0.
\end{equation}
Following the Ferrari method and introducing $\lambda=z-\frac{a_3}{4}$, the quartic equation \eqref{eq:quartic} can be reduced to the depressed quartic equation
\begin{equation}\label{eq:dep_quartic}
    z^4 + b_2 z^2 + b_1 z + b_0 = 0,
\end{equation}
where
\begin{equation*}
    b_2 = a_2 -\frac{3}{8} a_3^2; \quad
    b_1 = a_1 -\frac{1}{2} a_2a_3 + \frac{1}{8}a_3^3;  \quad
    b_0  = a_0 - \frac{1}{4} a_1a_3 + \frac{1}{16} a_2a_3^2 - \frac{3}{256} a_3^4.
\end{equation*}
It can be shown that if $w$ solves the cubic polynomial
\begin{equation}\label{eq:cubic_poly}
    w^3 + 2 b_2 w^2 + (b_2^2-4b_0)w - b_1^2 = 0,
\end{equation}
then the depressed quartic function in \eqref{eq:dep_quartic} can be decomposed as the product of two quadratic functions:
\begin{equation*}    
    \left(z^2 + \frac{b_2}{2} + \frac{w}{2}  + \sqrt{w} z - \frac{b_1}{2\sqrt{w}} \right) \cdot     \left(z^2 + \frac{b_2}{2} + \frac{w}{2}  - \sqrt{w} z + \frac{b_1}{2\sqrt{w}} \right) = 0.
\end{equation*}
Thus one can obtain the solutions of the quartic equation by solving the two quadratic equations above:
\begin{equation}\label{eq:root_quartic}
    \begin{aligned}
        \lambda_{1,0} & = \frac{a_3}{4} - \frac{\sqrt w}{2} + \frac{\sqrt{u_+}}{2},  \quad
        \lambda_{2,0} = \frac{a_3}{4} - \frac{\sqrt w}{2} - \frac{\sqrt{u_+}}{2}, \\
        \lambda_{3,0} & = \frac{a_3}{4} + \frac{\sqrt w}{2} + \frac{\sqrt{u_-}}{2},  \quad
        \lambda_{4,0} = \frac{a_3}{4} + \frac{\sqrt w}{2} - \frac{\sqrt{u_-}}{2}, \\        
    \end{aligned}
\end{equation}
where
\begin{equation*}
     u_\pm:=-2b_2-2w \pm 2b_1/\sqrt{w}.
\end{equation*}
In the above, $w$ can be expressed explicitly using the cubic formula:
\begin{equation}\label{eq:vw}
    w = -\frac{1}{3} \left( 2b_1 + v + \frac{\Delta_0}{v} \right), \quad v = \frac{1}{\sqrt[3]{2}} \cdot \sqrt[3]{\Delta_1 + \sqrt{\Delta_1-4\Delta_0^3}},
\end{equation}
where
\begin{equation*}
    \Delta_0 = a_2^2 - 3a_1a_3 + 12 a_0, \quad
    \Delta_1 = 2 a_2^3 - 9 a_1a_2a_3 + 27 a_0 a_3^2 + 27 a_1^2 - 72 a_0a_2.
\end{equation*}
In particular, $w$ above is chosen to be a real root of \eqref{eq:cubic_poly} when $a_j$ are real.
The discriminant of the quartic function is defined as $\Delta = -\frac{\Delta_1-4\Delta_0^2}{27}$.

\section{Transfer matrices for the constant media}\label{sec:transfer_mat}
It can be computed that the transfer matrix for the $A$ layer takes the following form (\cite{figotin2001}):
\begin{equation*}
   \cT_A(\omega; L) = \left(
  \begin{matrix}
        t^2\rho_1 + \tilde{t}^2\rho_2     &   t\tilde{t}(\rho_1 - \rho_2)    &  \i t\tilde{t}\left( \tilde\rho_2n_2^{-1} - \tilde\rho_1n_1^{-1} \right) &  \i t^2\tilde\rho_1 n_1^{-1} + \i \tilde{t}^2\tilde\rho_2 n_2^{-1}   \\
         t\tilde{t}(\rho_1 - \rho_2)  &  \tilde{t}^2\rho_1 + t^2\rho_2       &      -\left(\i t^2\tilde\rho_1n_1^{-1} + \i \tilde{t}^2\tilde\rho_2n_2^{-1} \right) &  \i t\tilde{t}\left(\tilde\rho_1n_1^{-1} - \tilde\rho_2n_2^{-1}\right) \\
        \i t\tilde{t}\left(\tilde\rho_2n_2 - \tilde\rho_1n_1 \right)  &   -(\i t^2\tilde\rho_1n_1 + \i \tilde{t}^2\tilde\rho_2n_2)    &   \tilde{t}^2\rho_1 + t^2\rho_2   &  t\tilde{t} (\rho_2 - \rho_1)     \\
        \i t^2\tilde\rho_1n_1 + \i \tilde{t}^2\tilde\rho_2n_2   & \i t\tilde{t}\left(\rho_1n_1 - \rho_2n_2 \right)  & t\tilde{t}(\rho_2 - \rho_1) & t^2\rho_1 + \tilde{t}^2\rho_2
  \end{matrix}
   \right),     
\end{equation*}
where
\begin{equation*}
    \begin{aligned}
       &  n_1 = \sqrt{\eps_0 + \delta}, \; n_2 = \sqrt{\eps_0 - \delta};  \\
       & \rho_1 = \cos( n_1 \omega L), \; \rho_2 = \cos( n_2 \omega L), \; t = \cos(\varphi); \\
       & \tilde\rho_1 = \sin( n_1 \omega L), \; \tilde\rho_2 = \sin( n_2 \omega L), \; \tilde{t} = \sin(\varphi).
    \end{aligned}
\end{equation*}
The transfer matrix for the $F$ layer is
\begin{equation*}
   \cT_F(\omega; L) =  \frac{1}{2} \left(
  \begin{matrix}
     \sigma_1 + \sigma_2    &  \i(\sigma_1 - \sigma_2)  &  \tilde\sigma_1 \tilde m_1^{-1} - \tilde\sigma_2 \tilde m_2^{-1}   &  \i\tilde\sigma_1 \tilde m_1^{-1} + \i \tilde\sigma_2 \tilde m_2^{-1}  \\
     \i(\sigma_2 - \sigma_1)   &  \sigma_1 + \sigma_2   &  -\i(\tilde\sigma_1 \tilde m_1^{-1} + \tilde\sigma_2 \tilde m_2^{-1})  &  \tilde\sigma_1 \tilde m_1^{-1} - \tilde\sigma_2 \tilde m_2^{-1}  \\
     \tilde m_2 \sigma_2 - \tilde m_1 \sigma_1  &  -\i (\tilde m_1 \sigma_1+\tilde m_2 \sigma_2)  &  \sigma_1 + \sigma_2    &  \i(\sigma_1 - \sigma_2)    \\
     \i(\tilde m_1 \sigma_1+\tilde m_2 \sigma_2) &  \tilde m_2 \sigma_2 - \tilde m_1 \sigma_1  &  \i(\sigma_2 - \sigma_1)   &  \sigma_1 + \sigma_2 
  \end{matrix}
   \right),     
\end{equation*}
where
\begin{equation*}
    \begin{aligned}
       &  m_1 = \sqrt{ (\tilde\eps_0 + \alpha)(1+\beta) }, \; m_2 = \sqrt{ (\tilde\eps_0 - \alpha)(1-\beta) };  \\
       &  \tilde m_1 = \sqrt{ (\tilde\eps_0 + \alpha)/(1+\beta) }, \; \tilde m_2 = \sqrt{ (\tilde\eps_0 - \alpha)/(1-\beta) };  \\
       & \sigma_1 = \cos( m_1 \omega L), \; \sigma_2 = \cos( m_2 \omega L); \\
       & \tilde\sigma_1 = \sin( m_1 \omega L), \; \tilde\sigma_2 = \sin( m_2 \omega L).
    \end{aligned}
\end{equation*}

\section{Another definition of the spectral bands for non-Hermitian photonic crystals}\label{sec:def_spec_bands}
Let $\xi=e^{\i k}\in\bbT$ be an eigenvalue of the monodromy matrix $\cM(\omega)$. In view of the characteristic equation \eqref{eq:chara1} and \eqref{eq:chara2}, the eigenfrequencies $\omega$ for the spectral problem \eqref{eq:eig_prob1} satisfy
\[
P(\omega, \xi):=\det\!\big(\cM(\omega)-\xi\,\cI\big)
=\xi^4+a_3(\omega)\xi^3+a_2(\omega)\xi^2+a_1(\omega)\xi+1=0 .
\]
Since $\cM(\omega)$ is entire in $\omega$, the function $P$, which is a polynomial in $\xi$, is holomorphic in $(\omega, \xi)\in\bbC\times\bbC^*$, where $\bbC^*:=\bbC \backslash\{0\}$. Define $\Sigma:=\big\{(\omega, \xi)\in\bbC\times\bbC^*: P(\omega,\xi)=0\big\}$. It is clear that the spectrum of the differential operator $\cL$ in \eqref{eq:eig_prob1} is given by
\[
\sigma(\cL):=
\big\{\omega\in\bbC: \exists\,\xi\in\bbT
\text{ such that }(\omega, \xi)\in\Sigma\big\}.
\]

For each fixed $\xi\in\bbT$, the entire function $\omega\mapsto P(\omega,\xi)$ has a discrete
zero set, which in the present setting is countably infinite; these zeros can be used to label the spectral
bands $n=1,2,\dots$. Now fix $(\xi_0,\omega_0)\in\Sigma$ with $\xi_0\in\bbT$. If $\partial_\omega P(\omega_0,\xi_0)\neq 0$,
the holomorphic implicit function theorem yields a unique holomorphic branch
$\omega=\varphi(\xi)$ in a complex neighborhood of $\xi_0$, with
$\varphi(\xi_0)=\omega_0$ and $(\xi,\varphi(\xi))\in\Sigma$. This gives the local analyticity of the spectral band. 

More precisely, let us introduce the exceptional set $\mathcal E:=
\big\{(\omega, \xi)\in\bbC\times\bbC^*:\partial_\omega P(\omega,\xi)=0\big\}$, which consists of points at which $\omega$ is a multiple zero of
$P(\cdot,\xi)$. Points in $\mathcal E_{\bbT} := \{ (\omega, \xi) \in \mathcal E: \xi \in \bbT \} $ are called exceptional points for the non-Hermitian eigenvalue problem \cite{bergholtz2021}. Near such a point, the local branches generally admit Puiseux expansions. We assume that the set $\mathcal E_{\bbT}$ is discrete.

To define the spectral bands $\{\omega_n(k)\}_{n=1}^\infty$,  we can fix one reference point $\xi_*\in\bbT$ such that all zeros of $P(\cdot,\xi_*)$ are simple. Then the discrete set $\{\omega\in\bbC:P(\omega,\xi_*)=0\}$ is enumerated as $\{\omega_n(\xi_*)\}_{n\geq 1}$ in any fixed manner. The branch $\omega_n(\xi)$ is
then defined as the holomorphic germ through $(\xi_*,\omega_n(\xi_*))$ and extended by analytic
continuation along paths in $\bbC^*\setminus\pi(\mathcal E)$, where
\[
\pi:\Sigma\to\bbC^*,\qquad \pi(\omega, \xi)=\xi,
\]
is the projection onto the $\xi$-plane. 
In this way, each $\omega_n(\xi)$ defines a locally holomorphic branch of the dispersion relation in $\bbC^*\setminus\pi(\mathcal E)$.

Finally, we emphasize that even if the analytic continuation of the branch $\omega_n(\cdot)$ is well defined along $\mathbb{T}$, it need not return to the original branch after one circuit of $\mathbb{T}$. Indeed,
the monodromy of $\Sigma$ may permute the branches, giving rise to the familiar non-Hermitian eigenvalue swapping around exceptional points, a phenomenon commonly
referred to as \emph{spectral braiding} \cite{zhang2023}. We refer to \cite[Ch.~8]{ahlfors1979complex} for the discussion on the analytic continuation.

\end{appendices}

\printbibliography
\end{document}